\documentclass[10pt]{amsart}
\usepackage{amsmath, amsfonts, amsthm, amssymb}
\usepackage[a4paper, width=16cm, top=30mm, bottom=30mm]{geometry}
\usepackage{natbib}
\usepackage{mathrsfs}
\usepackage{subfigure}
\usepackage{stmaryrd}
\usepackage{verbatim}
\usepackage{hyperref}
\usepackage{color}
\usepackage{undertilde}
\usepackage{bbm}
\usepackage{ulem}
\usepackage[dvips]{graphicx}
\usepackage{enumerate}
\numberwithin{equation}{section}
\newtheorem{theorem}{Theorem}[section]
\newtheorem{corollary}[theorem]{Corollary}
\newtheorem{lemma}[theorem]{Lemma}
\newtheorem{proposition}[theorem]{Proposition}

\theoremstyle{definition}
\newtheorem{definition}[theorem]{Definition}
\newtheorem{remark}[theorem]{Remark}
\newtheorem{assumption}[theorem]{Assumption}

\newcommand{\ind}{1\hspace{-2.1mm}{1}}

\newcommand{\ff}{\mathrm{f}}

\newcommand{\gr}{\mathfrak{g}}
\newcommand{\FF}{\mathbb{F}}
\newcommand{\RR}{\mathbb{R}}

\newcommand{\EE}{\mathbb{E}}

\newcommand{\Ll}{\mathcal{L}}

\newcommand{\Cc}{\mathcal{C}}
\newcommand{\Cb}{\mathfrak{C}}
\newcommand{\crbs}{\mathfrak{c}}

\newcommand{\df}{\mathfrak{d}}
\newcommand{\Vr}{\mathscr{V}}
\newcommand{\Mm}{\mathcal{M}}

\newcommand{\Fr}{\mathfrak{F}}
\newcommand{\Gr}{\mathfrak{G}}
\newcommand{\Lr}{\mathfrak{L}}
\newcommand{\Mf}{\mathfrak{M}}
\newcommand{\Wf}{\mathcal{W}}
\newcommand{\Mbf}{\mathbf{M}}
\newcommand{\Oo}{\mathcal{O}}
\newcommand{\Pp}{\mathcal{P}}
\newcommand{\Ii}{\mathcal{I}}
\newcommand{\Jj}{\mathcal{J}}
\newcommand{\Tt}{\mathcal{T}}
\newcommand{\NN}{\mathbb{N}}

\newcommand{\MMM}{\mathbb{M}}

\newcommand{\E}{\mathrm{e}}

\newcommand{\Nn}{\mathcal{N}}

\newcommand{\Ff}{\mathscr{F}}

\newcommand{\eps}{\varepsilon}

\newcommand{\Span}{\mathrm{Span}~}
\newcommand{\rd}{\ensuremath{\mathrm{d}}}

\newcommand{\Kk}{\mathcal{K}}
\newcommand{\Hh}{\mathcal{H}}

\newcommand{\Aa}{\mathcal{A}}

\newcommand{\Bb}{\mathcal{B}}
\newcommand{\Ss}{\mathcal{S}}
\newcommand{\cc}{\mathrm{c}}
\newcommand{\ww}{\mathrm{w}}

\newcommand{\BS}{\mathrm{BS}}
\newcommand{\CC}{\mathrm{C}}

\newcommand{\Tto}{\mathcal{T}_0}

\newcommand{\Kf}{\mathfrak{K}}

\newcommand{\hh}{\mathrm{h}}
\newcommand{\uu}{\mathrm{u}}
\newcommand{\vv}{\mathrm{v}}

\newcommand{\Xx}{\mathcal{X}}
\newcommand{\Yy}{\mathcal{Y}}

\newcommand{\barho}{\overline{\rho}}

\newcommand{\val}{\vartheta}

\newcommand{\Uu}{\mathcal{U}}
\newcommand{\interior}{\mathrm{int}~}
\newcommand{\aaa}{\mathrm{a}}
\newcommand{\sss}{\mathfrak{p}}

\begin{document}
\title{Perturbation analysis of sub/super hedging problems}

\author{Sergey Badikov}
\address{Department of Mathematics, Imperial College London}
\email{sergey.badikov08@imperial.ac.uk}

\author{Mark H.A. Davis}
\address{Department of Mathematics, Imperial College London}
\email{mark.davis@imperial.ac.uk}

\author{Antoine Jacquier}
\address{Department of Mathematics, Imperial College London, and the Alan Turing Institute}
\email{a.jacquier@imperial.ac.uk}

\date{\today}

\thanks{The authors are indebted to the anonymous referee and the Associate Editors for their insightful comments.
}
\begin{abstract}
We investigate the links between various no-arbitrage conditions and the existence of pricing functionals 
in general markets, and prove the Fundamental Theorem of Asset Pricing therein.
No-arbitrage conditions, either in this abstract setting or in the case of a market consisting of European Call options, 
give rise to duality properties of infinite-dimensional sub- and super-hedging problems.
With a view towards applications, we show how duality is preserved when reducing 
these problems over finite-dimensional bases.
We also introduce a rigorous perturbation analysis of these linear programming problems, 
and highlight numerically the influence of smile extrapolation on the bounds of exotic options.
\end{abstract}
\maketitle

\section{Introduction}
In mathematical finance, pricing contingent claims consists in postulating the existence of a filtered probability space (or of a model, using the terminology in~\cite{CoxObloj11}) 
such that the discounted price process is a martingale. 
In the absence of arbitrage (appropriately defined), prices of claims 
can then be expressed as expectations of the discounted payoffs under a martingale measure. 
The postulated model is in general not unique, 
and a whole range of prices arises as all possible models are taken into account, 
together with no-arbitrage constraints. 
In contrast, model-independent finance strives to move away from this paradigm, 
and instead relies on no-arbitrage conditions and additional market information 
to find arbitrage-free bounds on prices of contingent claims. 

\cite{Hob98} posited no model at all and instead used no-arbitrage assumptions to derive arbitrage-free range of possible prices for exotic derivatives.
This approach fundamentally relies on the Skorokhod embeddings
and Dambis-Dubins-Schwarz time-change techniques,
and a vast literature on arbitrage-free bounds on prices of derivatives has grown 
since~\cite{{BonnansTan, Brown, CoxHoeg, CoxObloj11, CoxWang, DavisOblojRaval14, Klimmek, Laurence, Neuberger, Ulmer}}.
More recently, this problem has been tackled using martingale optimal transportation theory,
first initiated by~\cite{BHP13},
who showed that when full marginals (equivalently all European Call/Put options) are known,
the problem of finding arbitrage-free bounds on prices of exotic derivatives
can be formulated as a martingale version of a Monge-Kantorovich mass transport problem. 
From a practical point of view, the appeal is that this formulation can be seen as an infinite-dimensional linear programming problem, with a dual that can be interpreted in terms of semi-static hedging strategies.
This seminal paper has since been extended to the case of finitely many marginals~\cite{DMHL15, GTT16, OS17}
and some of its technical assumptions, either on the marginals or on the cost function 
to be minimised, have been relaxed~\cite{BJ16, HLT16, HLXT16}.
An underlying question is whether observed option prices yield any kind of arbitrage in the market.
This relation between market data and fundamental theorem of asset pricing has been made precise, 
in the model-independent framework, by~\cite{ABPS13, Bayraktar, Cousot, DH07}.
\cite{BN15} formulate the fundamental theorem and the superhedging problem in the quasi-sure setting, 
where all statements hold outside polar sets of a collection of probability measures~$\Pp$, not necessarily equivalent,  
on the measurable state space $(\Omega,\Aa)$ that governs the market. 
They obtain the first fundamental theorem and the superhedging property in a multi-period setting with possible inclusion of options for static hedging. 
They show in particular~\cite[Example 1.2]{BN15} that if~$\Pp$ is the set of all Borel probability measures on a finite $d$-dimensional state space then the quasi-sure inequalities become pointwise.

In this paper, we first investigate in Section~\ref{sec:generalDuality} the relations between absence of weak free lunch, similar to the `free lunch' introduced by \cite{Kreps81},
and the existence (and extension) of pricing functionals in general abstract markets. 
In order to represent the extension as a Borel probability measure on a locally compact state space, 
we assume the existence of a strictly positive continuous function dominating the payoffs 
of the traded assets along with a technical assumption. 
We further show  how to sub/super-replicate general options in this general market.
We then (Section~\ref{sec:ApplicationToOptionPricesCh2}) specialise the market to the case 
where only finitely many European Call options are traded for a given set of maturities. 
In order to avoid the emergence of duality gaps, we introduce a restriction on the set of feasible dual solutions in the form of the total implied variance extrapolation and how such extrapolation is connected to existence of higher moments of the underlying price process (Section~\ref{sec:extrapVarCh2}).
We discuss in Section~\ref{sec:CaseOfFinitelyManyOptionsCh2} the notion of weak arbitrage, 
introduced by \cite{DH07}, which naturally leads to the introduction of sub- and super-hedging problems.
We finally investigate the impact of the extrapolation of the total implied variance on the latter:
we first discretise the latter to obtain semi-infinite linear programmes (Section~\ref{sec:semiInfiniteLP}), 
and prove convergence as the discretisation becomes finer. 
Section~\ref{sec:PerturbationAnalysis} is devoted to a perturbation analysis, following~\cite{BonnansShapiro00},
of the initial inputs (Call option prices) in the optimisation problem, which provides the user with a better control over model parameters and extrapolation issues.
We illustrate numerically our findings in several examples common in Finance 
in Section~\ref{sec:applicationFwdStartStraddleCh4}.

\section{Preliminary results}\label{sec:generalDuality}

We establish super-hedging duality in general markets as an application of infinite-dimensional linear programming. 
The general market consists of securities with continuous
payoffs~$(\varphi_i)_{i\in \Ii}$ 
and traded at prices~$(c_i)_{i\in \Ii}$, with~$\Ii$ some index set.
Since the market is assumed frictionless, the set of traded securities becomes a subspace 
of the space of continuous functions, on which we introduce
a pricing functional mapping payoffs of traded securities to market prices. 
We fix an index set~$\Ii$ (not necessarily finite) and a collection of functions
$\varphi_i\in \Cc(\Omega)$, $i\in\Ii$ representing payoffs of securities available on the market 
at finite prices $c_i\in\RR$. 
We assume that the market is frictionless, i.e. there are no transaction costs associated with buying and selling securities, 
there are no liquidity constraints and market participants are allowed to buy and sell any position in a security or a portfolio of securities.
Denote by~$\Mf$ the space of traded claims, i.e. the set of portfolios of securities that can be bought and sold freely on the market, as 
\begin{equation}\label{eq:setOfTradedClaimsCh1}
\Mf := \left\{\sum_{n=1}^N\alpha_n\varphi_{i_n}: (\alpha_n)_{n=1,\ldots,N}\in\RR^N, N\in\NN \text{ and } i_1,\ldots,i_N\in\Ii\right\}.
\end{equation}
As trading is frictionless, $\Mf$ is a linear subspace of $\Cc(\Omega)$. 
Define also a pricing functional $\rho : \Mf \to \RR$ mapping payoffs to their market prices 
\begin{equation}\label{eq:pricingOperator}
\rho(m) := \left\{\sum_{n=1}^N\alpha_nc_{i_n} : 
m = \sum_{n=1}^N\alpha_n\varphi_{i_n} 
\text{ for some } N\in\NN \text{, } i_1,\ldots,i_N\in\Ii\right\}.
\end{equation}
Although it is defined as a set-valued function, below we show that absence of arbitrage is equivalent to certain properties of the pricing functional, including being single valued.
Before we proceed we make a regularity assumption on the market 
that will allow us to establish separating duality in the sequel. 
\begin{assumption}\label{assumption:functionH}
There exists a reference claim, namely a continuous function $h : \Omega \to \RR_+\cup\{\infty\}$ with the following properties: 
\begin{enumerate}
\item $h$ has compact  level sets ($\{\omega: h(\omega)\leq K\}$ is compact for all $K>0$) and $1/h$ is bounded on~$\Omega$
(there exists $C>0$ such that $\{ 1/C \leq 1/h(\omega) \leq C \text{ for all }\omega\in\Omega\}$);
\item $h\notin \Mf$, i.e. $h$ is not a tradeable asset;
\item $\varphi_i = o(h)$ (as $\|\omega\|_1\uparrow\infty$) for all $i\in\Ii$.
\end{enumerate}
\end{assumption}

Assumption~\ref{assumption:functionH} has already been considered in the literature, 
albeit with slightly different flavours.
\cite{CKT17} assume existence of a continuous function 
$h: \Omega \to [1,+\infty)$, 
with bounded level sets $\{h^{\leftarrow}(-\infty, z): z\in\RR_+\}$, 
and consider payoffs that are upper and lower semi-continuous and bounded by~$h$.
Under additional technical existence assumptions, they allow for claims growing at most linearly, 
extending the results by \cite{ABPS13}. 
The latter indeed assume existence of a super-linear convex function, 
accounting for the pay-off of a traded option
(equivalent to assuming infinitely many traded European Call options).
Their assumptions on~$h$ are weaker than ours, requiring bounded level sets as opposed to compact level sets. 
However, our setting is more general since we allow for the more realistic case case of finitely many options traded on the market.
We mention in passing the works by \cite{BN15} and by \cite{BFM16}, using a quasi-sure approach:
they replace the pathwise superhedging inequality with an inequality that holds outside `maximal polar sets' common to a set of non-dominated probability measures. 
This however is a different route than ours, and we refer the interested reader to these papers for more details.
Here and elsewhere, $\RR_{+}:=[0,\infty)$ denotes the non-negative half-line.
With the weighted space 
\begin{equation}
\Cc_{h}(\Omega)
 := \left\{ f\in \Cc(\Omega) : \|f\|_{h}:=\sup_{\omega\in \Omega}\frac{|f(\omega)|}{h(\omega)}<\infty\right\},
\end{equation}
Assumption~\ref{assumption:functionH} implies that $\Mf\subset \Cc_{h}(\Omega)$,
and the topology on~$\Mf$ is the one inherited from~$\Cc_{h}(\Omega)$.
Endowed with~$\|\cdot\|_h$, $\Cc_{h}(\Omega)$ is a Banach lattice, 
and the order unit in $\Cc_{h}(\Omega)$ is~$h$ 
(Definitions~\ref{defn:OrderUnit} and~\ref{defn:BanachLattice}).
Following arguments from~\cite[Example 8.6.5]{Bogachev}, 
the topological dual of~$\Cc_{h}(\Omega)$ is the space of signed Borel measures 
that integrate~$h$ to a finite constant:
\begin{equation}
(\Mm_h)_+(\Omega) := \left\{\mu \in \Mm_+(\Omega) : \langle h, \mu\rangle < \infty\right\},
\end{equation}
with~$\Mm(\Omega)$ the set of signed Borel measures on~$\Omega$ 
(the notations $\Mm_+$ and~$\Mm_{++}$ are introduced in Definition~\ref{defn:cones} )
and the bilinear form
\begin{equation}
\left \langle f, \mu\right\rangle := \int_{\Omega}f(\omega)\mu(\rd\omega),
\quad\text{for all }f\in \Cc_{h}(\Omega), \mu\in \Mm_h(\Omega).
\end{equation} 
If the total variation of a measure $\mu\in(\Mm_h)_+(\Omega)$ is equal to one then $\mu\in\Pp_h(\Omega)$, 
where~$\Pp_h(\Omega)$ denotes the set of Borel probability measures that integrate $h$ to a finite constant. 
We now define a notion of arbitrage in this abstract market, using notation introduced in Definition~\ref{defn:cones}.
\begin{definition}\label{dfn:MIArb}
There is no strong model-independent arbitrage on~$\Mf$ 
if $\inf\rho(m)\geq 0$ for all $m\in \Mf_+$,
and $\inf\rho(m)>0$ for all $m\in \Mf_{++}$.
\end{definition}
This definition is inspired by, yet stronger than, 
that of absence of model-independent arbitrage in~\cite[Definition 2.1]{DH07},
which holds if $\rho(m) \geq 0$ for all $m\in \Mf_+$.
In order to avoid the degenerate situation $\rho(m) = 0$ for all $m\in \Mf_+$ we make the following assumption: 
\begin{assumption}\label{assumption:risklessBond}
There exists a traded claim $m_0\in \Mf$ with $m_0(\omega) >0$ for all $\omega\in \Omega$ 
and $\rho(m_0) > 0$.
\end{assumption}
Assumption~\ref{assumption:risklessBond} holds 
if a riskless bond is available on the market
and implies that the two statements in Definition~\ref{dfn:MIArb} are equivalent.
In general~$\rho$ is a set-valued function, but the following restricts its range:
\begin{proposition}\label{thm:rhoProperties}\cite[Theorem~3]{Clark93}
Under Assumption~\ref{assumption:risklessBond},
absence of strong model-independent arbitrage holds if and only if~$\rho$, 
defined in~\eqref{eq:pricingOperator}, is strictly positive, linear and uniquely defined. 
\end{proposition}
An earlier version of this theorem for Ross' No Arbitrage was proved by \cite{Kreps81}.
Let us define the set of feasible claims, i.e. traded claims available at non-positive prices, as 
\begin{equation}\label{eq:setF}
\Fr := \{m\in \Mf: \inf\rho(m) \leq 0\}.
\end{equation} 
Ross' principle of no-arbitrage~\cite{Ross} in the consumption space~$\mathrm{L}$
reads~\cite{Clark93} $\Fr\cap \mathrm{L}_{++}(\Omega) = \emptyset$,
where~$\mathrm{L}$ is a set of random variables with reference to a given probability measure.
Under Assumption~\ref{assumption:risklessBond}, $\rho(0)=0$,
since 
$\mathrm{L}_{+} = \mathrm{L}_{++} \cup (\mathrm{L}_{+}\setminus \mathrm{L}_{++})$, 
this is equivalent to 
$\Fr\cap \mathrm{L}_{+} = \{0\}$.
This is clearly equivalent to Definition~\ref{dfn:MIArb}.
It is however different from Stricker's No Approximate Arbitrage principle~\cite{Stricker} 
$\overline{\Fr}\cap \mathrm{L}_{++}(\Omega) = \emptyset$, 
which involves the closure with respect to the weak topology on~$\mathrm{L}$.
Our framework follows the model-independent approach, without reference to a given probability measure.
Proposition~\ref{thm:rhoProperties} implies that $\rho(0) = 0$,
and the following representation of~$\Mf$ holds:
\begin{lemma}
Under Assumption~\ref{assumption:risklessBond}, $\Mf = \Span\{m_0,\Fr\}$.
\end{lemma}
\begin{proof}
It is immediate to see that $\Span\{m_0,\Fr\}\subseteq \Mf$. 
On the other hand for any $m\in \Mf$ available at price $\rho(m)$ 
define $f:=m-[\rho(m)/\rho(m_0)]m_0$ with $\rho(f) = 0$ and thus $f\in \Fr$. 
Then~$m$ can be trivially represented as a linear combination $f + [\rho(m)/\rho(m_0)]m_0$
and the reverse inclusion follows.
\end{proof}
Although the notion of strong model independent arbitrage is helpful to restrict the range of $\rho$, it does not appear to be helpful in explaining the behaviour of the pricing functional on the closure of the feasible set of claims. 
The following notion of arbitrage is similar in flavour to~\cite[Definition~2.1]{CoxObloj11}:
\begin{definition}\label{defn:WFL}
There is a weak free lunch if there exists a sequence  
$(g_n)_{n\in\NN}\subset \Cc_{h}(\Omega)$ 
converging weakly to $g\in(\Cc_{h})_{++}(\Omega)$,
and a sequence $(f_n)_{n\in\NN}\subset \Fr$ with $f_n\geq g_n$ for all $n\in\NN$.
\end{definition}
It must be noted that a strong model-independent arbitrage is also a week free lunch.
Before we proceed let us first show an auxiliary result.
\begin{lemma}\label{lemma:AlgebraicDiffConeCh1}
The following equality holds for the algebraic difference $\Fr-(\Cc_{h})_+(\Omega)$:
$$
\Fr-(\Cc_{h})_+(\Omega) 
:= \{f-g : f\in \Fr, g\in (\Cc_{h})_+(\Omega)\}
  = \left\{g\in \Cc_{h}(\Omega) : \text{ there exists }f\in \Fr\text{ such that }f\geq g\right\}
  =:\Gr.
$$
\end{lemma}
\begin{proof}
For any $g\in \Gr$ there exists $f\in \Fr$ such that $f-g\in (\Cc_{h})_+(\Omega) $ or equivalently $g-f\leq 0$.
As $0\in \Fr$ we have that $0-(f-g) \in \Fr-(\Cc_{h})_+(\Omega)$ hence $\Gr\subseteq \Fr-(\Cc_{h})_+(\Omega)$. 
On the other hand let $f\in \Fr$ and $z\in (\Cc_{h})_+(\Omega)$.
Let $g:= f-z$ and note that $f\geq g$. Hence $g\in \Gr$ and it follows that  $\Fr-(\Cc_{h})_+(\Omega) \subseteq \Gr$.
\end{proof}
Lemma~\ref{lemma:AlgebraicDiffConeCh1} still applies if the positive cone $(\Cc_{h})_+(\Omega)$ 
is restricted to $\Mf_+$.
It follows that the absence of weak free lunch
can equivalently be stated as $\overline{\Fr-(\Cc_{h})_+(\Omega)}\cap (\Cc_{h})_{+}(\Omega) = \{0\}$, 
where the closure is taken with respect to the weak topology on $\Cc_{h}(\Omega)$.
We are now ready to state a version of the Fundamental theorem, proved in Appendix~\ref{sec:thmFTAPnoOrderUnit}:
\begin{theorem}\label{thm:FTAPnoOrderUnit}
Under Assumptions~\ref{assumption:functionH} and \ref{assumption:risklessBond},
absence of weak free lunch holds if and only if there exists a continuous strictly positive linear functional 
$\pi: \Cc_{h}(\Omega) \to \RR$ that extends $\rho$. 
Moreover $\pi$ can be written as an integral with respect to a unique Borel probability measure $\mu\in\Pp_h(\Omega)$.
\end{theorem}
For a sequence $(m_n)_{n\in\NN}\subset \Mf$ converging weakly to $m\in\overline{\Mf}$ (the weak closure of~$\Mf$) 
define $\barho(m):=\lim_{n\uparrow\infty}\rho(m_n)$. 
It can be shown that $\barho$ is continuous, strictly positive and linear as a simple corollary to Theorem~\ref{thm:FTAPnoOrderUnit}.
We now formulate the super- and sub-hedging problems as infinite-dimensional linear programming problems. 
The super-hedging problem for an option with payoff 
$\Phi\in \Uu_h(\Omega)$,
the set of upper semi-continuous functions bounded by~$h$, is formulated as
\begin{equation}\label{eq:primal}
\overline{\val}_p(\Phi) := \inf\left\{ \barho(m) : m\in\overline{\Mf}, m(\omega)\geq\Phi(\omega)\text{, for all }\omega\in\Omega \right\}.
\end{equation}
The dual problem consists in finding a Borel probability measure subject to market constraints maximising (minimising in case of sub-hedging) the price of a derivative to be hedged and is stated as follows:
\begin{equation}\label{eq:dual}
\overline{\val}_d(\Phi) := \sup\left\{\langle \Phi, \mu\rangle: \mu\in\Pp_h(\Omega), \langle m, \mu\rangle = \barho(m), m\in\overline{\Mf}\right\}.
\end{equation}
We define here the sub/super-hedging problems in terms of the extension~$\barho$ instead of~$\rho$ 
itself as continuity of the former is essential for duality purposes.
The sub-hedging problem for an option with payoff 
$\Phi\in \Ll_h(\Omega)$,
the set of lower semi-continuous functions bounded by~$h$, can be stated as 
\begin{equation}\label{eq:subhedgePrimal}
\underline{\val}_p(\Phi) := \sup\left\{ \barho(m) : m\in\overline{\Mf}, m(\omega)\leq\Phi(\omega)\text{, for all }\omega\in\Omega \right\},
\end{equation}
and its dual problem is written as follows
\begin{equation}\label{eq:subhedgeDual}
\underline{\val}_d(\Phi) := \inf\left\{\langle \Phi, \mu\rangle: \mu\in\Pp_h(\Omega), \langle m, \mu\rangle = \barho(m), m\in\overline{\Mf}\right\}.
\end{equation}
It is easily seen that weak duality
$\underline{\val}_p(\Phi) \leq \underline{\val}_d(\Phi) \leq \overline{\val}_d(\Phi) \leq \overline{\val}_p(\Phi)$
holds,
at least for $\Phi \in \Ll_h(\Omega) \cap \Uu_h(\Omega)$.
As~$h$ is not a traded asset, the following assumption prevents degeneracy of the primal problem~\eqref{eq:primal}:
\begin{assumption}\label{assumption:primalfeasibility}
For any fixed $\Phi\in \Uu_h(\Omega)$, 
there exists $m\in\overline{\Mf}$ such that $m\geq\Phi$ on~$\Omega$.
\end{assumption}
The assumption implies that $\overline{\val}_p(\Phi)$ is feasible for any $\Phi\in \Uu_h(\Omega)$; 
since~$\barho$ is continuous on~$\overline{\Mf}$, it is also finite.
Since $\underline{\val}_p(-\Phi) =-\overline{\val}_p(\Phi)$, 
the sub-hedging problem~\eqref{eq:subhedgePrimal} is feasible for~$\Phi$ if~$-\Phi$ satisfies Assumption~\ref{assumption:primalfeasibility}.
The following result, proved in Appendix~\ref{app:thmsuperReplication},
provides absence of duality gap between the primal and dual problems.

\begin{theorem}\label{thm:superReplication}
Suppose  Assumptions~\ref{assumption:functionH}, \ref{assumption:risklessBond} and \ref{assumption:primalfeasibility} hold. 
Then absence of weak free lunch implies no duality gap between the primal and dual super-hedging (resp. sub-hedging) problems.
\end{theorem}

\section{Duality in markets with Call options}\label{sec:ApplicationToOptionPricesCh2}
We now consider when European Call options are traded on the market and discuss how Assumption~\ref{assumption:functionH} can be represented via arbitrage-free extrapolation of the total variance and how it affects the set of feasible solutions to the primal and dual problems. 
We also introduce another
notion of arbitrage to deal with specific cases of Butterfly option spreads priced at zero, allowing us to consider positive rather than strictly positive pricing functionals. 
Relaxing this condition also requires a different ordering on the cone of feasible claims and we show duality results still hold with the latter.

\subsection{Market definitions}
We work in a discrete time setting with a finite time horizon~$T$ and intermediate times $0=t_0<t_1<\ldots<t_n=T$.
The collection of times is defined to be $\Tto:=\{t_0,t_1,\ldots,t_n\}$,
and $\Tt := \Tto\setminus\{t_0\}$.
The state space $\Omega:=\prod_{t\in\Tt} \Omega_t$, where $\Omega_t:=\RR_+$, is locally compact, 
and the coordinate process $S:\Omega\to\RR_+$ is defined to be $S_{t}(\omega) = \omega_{t}$ 
for all $\omega\in\Omega$ and $\omega_t\in\Omega_t$.
We also  normalise it so that $S_0(\omega) = s_{0} = 1$.
We assume that for each maturity $t\in\Tt$, there are 
European Call options traded on the market at the price $c(K,t)$, 
with forward moneyness~$K$ in a finite set~$\Kf_t$.
We also refer to forward log-moneyness $k=\log(K)$, and we shall interchangeably use $c(k,t)$ and $c(K,t)$.
Let us define $K^t_*$ for each $t\in\Tt$ as the moneyness of a Call option available on the market at zero cost:
\begin{equation} 
K^t_*:=\inf\{K\in\Kf_t : c(K,t) = 0\}, 
\end{equation}
and $K^t_* = \infty$ if the set is empty. Denote by $\Cb:=\{c(K,t) : K\in\Kf_t, t\in\Tt\}$ the collection of prices of traded Call options.
\begin{definition}\label{defn:staticPosition}
A static position $\ff := (\varphi_t)_{t\in\Tto}$ is a collection of maps
from~$\RR\to\RR$, with $\varphi_{t_0} \in\RR$ such that, 
for each $t\in\Tt$, there exists $(\alpha_i)_{i=1,\ldots,\kappa(t)}\in\RR^{\kappa(t)}$,
$K^t_1,\ldots,K^t_{\kappa(t)}\in \Kf_t$, with $\kappa(t)<\infty$, for which
$$
\varphi_t := \sum_{i=1}^{\kappa(t)}\alpha_i(S_t - K^t_i)_+.
$$
\end{definition}
This function represents the payoff of the static position, with price at inception
$c_t := \sum_{i=1}^{\kappa(t)}\alpha_ic(K^t_{i},t)$,
and $\varphi_{t_0}$ a static position in a riskless bond with unit payoff.
The set of all static positions is denoted~$\Ss$.
\begin{definition}\label{def:tradingStrategy}
A trading strategy is a vector 
$\Delta := \left(\Delta_t\right)_{t=t_0,\ldots,t_{n-1}} \in \Hh$, 
where $\Hh:=\RR\times\prod_{j=1}^{n-1}\Cc_{b}(\RR^j_+)$ denotes the set of trading strategies.
The first component denotes the initial position in the stock 
and the other components are continuous and bounded functions.
The stochastic integral is defined as
$$
\left(\Delta \bullet S(\omega)\right)_T :=\sum_{i=0}^{n-1}\Delta_{t_i}(\omega)\left(S_{t_{i+1}}(\omega)-S_{t_{i}}(\omega)\right),
$$
and represents the gains or losses obtained by trading according to $\Delta$. 
We use notation $\Delta_{t_i}(\omega) := \Delta_{t_i}(\Pr\omega)$, 
where $\Pr\omega$ is the projection of $\omega\in\Omega$ onto $\RR^i_+$ for each $i=1,\ldots,n-1$.
\end{definition}
At time~$t_j$ (for $j=1,\ldots,n-1$), we consider the strategy~$\Delta_{t_j}$ 
as an element of~$\Cc_{b}(\RR^j_+)$. 
This takes into account possible absence of Markovianity of the underlying price process
or European options with path-dependent payoffs,
in which case the trading strategy depends, not only on the current value, but on the whole history of the price process.
The above definition includes the trivial strategy 
$\widetilde{\Delta} = \left(1,1,\ldots,1,1\right)$
of entering a forward contract at time zero maturing at~$T$
(or equivalently entering a forward contract with maturity~$t_1$ and rolling it to the final maturity~$T$), 
with payoff
$(\widetilde{\Delta}\bullet S(\omega))_T = S_T(\omega)-1$ for all $\omega\in\Omega$.
Also note that the payoff of any trading strategy $\Delta\in\Hh$ is at most linear in~$\omega$.
For a static position $\ff\in\Ss$ and a trading strategy $\Delta \in \Hh$, 
we write the initial cost and final payoff of a semi-static portfolio $(\ff,\Delta)$ as
\begin{equation}\label{eq:NotPi0T}
\Pi_{t_0}(\ff, \Delta) := \varphi_{t_0} + \sum_{t\in\Tt} c_{t}
\qquad\text{and}\qquad
\Pi_T(\ff, \Delta; \omega) := \varphi_{t_0} + \sum_{t\in\Tt}\varphi_{t}(S_t(\omega)) + (\Delta\bullet S(\omega))_T,
\end{equation}
for all $\omega\in\Omega$.
Note that it is possible to have a semi-static portfolio with final maturity $t<T$. 
However as we work with normalised prices, one can represent the final payoff of a portfolio 
maturing at time $t<T$ as a position in the riskless bond maturing at~$T$ 
with the value of the position equal to the said payoff.
The set of traded claims~$\Mf$ is then defined as a collection of all semi-static portfolio payoffs 
$\Pi_T(\ff, \Delta; \cdot)$ for a static position $\ff\in\Ss$ and a trading strategy $\Delta\in\Hh$,
\begin{equation}\label{eqn:market}
\Mf = \left\{\Pi_T(\ff, \Delta; \cdot) : \ff\in\Ss \text{ and } \Delta\in\Hh \right\}.
\end{equation}
As we assume that only European Call options are traded for each maturity $t\in\Tt$ and the payoff of a trading strategy $\Delta\in\Hh$ is continuous and grows at most linearly in $\omega\in\Omega$, 
the set~$\Mf$ consists of functions $m\in \Cc(\Omega)$ such that 
$m(\omega) = \Oo(1+\|\omega\|_1)$ as $\|\omega\|_1$ tends to infinity. 
It is in fact a subspace of~$\Cc_{l}(\Omega)$ where
\begin{equation}\label{eq:functionHLinear}
l(\omega) := 1 + \sum_{t\in\Tt_0} S_{t}(\omega).
\end{equation}
Note that $l\in \Mf$, as the semi-static portfolio $(\ff_*,\Delta_*)$ with $\ff_*:= (2+n,0,\ldots,0)$ and $\Delta_* := (n,n-1,\ldots,1)$
has final payoff $\Pi_T(\ff_*,\Delta_*;\cdot) = l$ on~$\Omega$.
The dual space is $\Pp_l(\Omega) := \{\mu\in\Pp(\Omega) : \langle l, \mu\rangle < \infty\}$, 
the space of all Borel probability measures with finite first moments.
Define now the pricing functional $\rho : \Mf \to \RR$ as 
\begin{equation}\label{eq:rhoFunctional}
\rho(\Pi_T(\ff, \Delta; \cdot)) := \Pi_{t_0}(\ff, \Delta).
\end{equation}
As above, Theorem~3 in~\cite{Clark93} implies that 
absence of strong model-independent arbitrage is equivalent to~$\rho$ 
being linear, uniquely defined and strictly positive.
We also define a market model similarly to~\cite[Definition 1.1]{CoxObloj11}.
\begin{definition}\label{def:Model}
A model is a probability measure in~$\Pp_l(\Omega)$ such that the coordinate process~$S$ is a martingale in its own filtration $\FF := (\sigma(S_r, r\leq t))_{t\in\Tto}$. 
A market model is a martingale measure associated with a positive linear extension of the pricing operator~$\rho$ (defined in~\eqref{eq:rhoFunctional}) from~$\Mf$ to~$\Cc_{l}(\Omega)$.
\end{definition}
A sufficient condition to ensure that~$S$ is a martingale under $\mu\in\MMM$ (the set of all martingale measures) 
is 
$\langle\left(\Delta \bullet S\right)_T, \mu\rangle = 0$,
for all $\Delta\in\Hh$.
By definition~$S$ is a martingale in its own filtration~$\FF$ under a measure $\mu\in\Pp_l(\Omega)$ if 
$\sum_{i=0}^{n-1}\left\langle\ind_{B_{t_i}}(\cdot)(S_{t_{i+1}} - S_{t_i}), \mu\right\rangle = 0$,
for all Borel sets $B_{t_i}\subset\Omega_{t_i}$  for all $i=1,\ldots,n-1$.
To see the sufficiency of the martingale condition, note that the Borel $\sigma$-algebra 
is generated by open sets and the indicator function of an open set is a lower semi-continuous function.
By Lebesgue Monotone Convergence Theorem the definition of a martingale follows.
Let us define the set of market models as 
\begin{equation}\label{eq:marketModels}
\MMM_{\Cb} := \Big\{\mu\in \Pp_l(\Omega): 
\langle \Pi_T(\ff,\Delta;\cdot), \mu\rangle = \Pi_{t_0}(\ff,\Delta)\text{ for }(\ff,\Delta)\in\Ss\times\Hh \Big\}.
\end{equation}
Here, $\Cb$ appears in the definition through $\Pi_{t_0}$ defined in~\eqref{eq:NotPi0T},
where~$c_t$ is the sum of elements of~$\Cb$. 
To enforce the existence of the function~$h$ we restrict the set of market models~$\MMM_{\Cb}$ by imposing conditions on existence of moments of the underlying stock process $S$ which are equivalent to allowing arbitrage-free extrapolation of the total implied variance as will be shown in the sequel. 

\subsection{Extrapolation of variance}\label{sec:extrapVarCh2}
We propose to restrict the set of market models~$\MMM_{\Cb}$  
by imposing conditions on arbitrage-free extrapolation of the total implied variance. 
The Black-Scholes formula for the arbitrage-free price of a Call option at time zero reads 
$c_{\BS}(k,\sigma\sqrt{t}) := \EE\{(S_t - \E^k)_+\} = \Nn(d) - \E^k\Nn(d - \sigma\sqrt{t})$, 
with $d := -\frac{k}{\sigma\sqrt{t}} + \frac{1}{2}\sigma\sqrt{t}$, 
where~$\Nn$ is the standard Normal distribution function. 
For a given market or model price $c(k,t)$ with log-moneyness~$k$ and maturity~$t$, 
the implied volatility 
$\sigma_{\mathrm{implied}}(k,t)$ is the unique non-negative solution to 
$c(k,t) = c_{\BS}(k,\sigma_{\mathrm{implied}}(k,t)\sqrt{t})$ and the total implied variance is then
$w(k,t) := \sigma_{\mathrm{implied}}^2(k, t)t$.
In practice only finitely many option prices are quoted on the market and hence the total implied variance 
function cannot be uniquely specified based on market quotes alone. 
We concentrate our attention on extrapolation of the total implied variance
for a fixed maturity~$t$, while preserving absence of arbitrage.
\cite{Lee} proved that a slice of the total variance~$k\mapsto w(k, t)$ 
can be at most linear as~$|k|$ tends to infinity, and related precisely the slope of the wings to 
the moments of the underlying stock price process.
\cite{BenaimFriz09, BenaimFriz08} further refined this analysis under additional conditions 
on the moment generating function of the log-returns distribution.
Absence of strong model-independent arbitrage (Definition~\ref{dfn:MIArb}) 
in presence of options is equivalent to absence of Calendar and Butterfly Spread arbitrages, 
understood as absence of arbitrage opportunities across option maturities for a fixed strike 
and absence of arbitrage opportunities across different strikes for a fixed maturity respectively.
We shall work with the following standing assumption on the total implied variance:
\begin{assumption}\label{assumption:totalVarConditions}
For fixed $k\in\RR$, $w(k,\cdot)\in \Cc^1(\RR_+)$. 
For fixed $t>0$, $w(\cdot,t)\in \Cc(\RR)$ is strictly positive, differentiable except possibly at finitely
many points, and $\partial_kw(k,t)$ is essentially bounded measurable. 
\end{assumption}
Absence of arbitrage can equivalently be stated as conditions on the shape of the total implied variance
 as shown in~\cite{GatheralJacquier14, JGMN16}.
In particular under proportional dividends, absence of Calendar Spread arbitrage is equivalent to
$\partial_tw(k,t)\geq 0$ for all $k\in\RR$ and $t>0$~\cite[Lemma 2.1]{GatheralJacquier14}.
This is equivalent to the Call price surface being non-decreasing in maturity for each strike.
For fixed~$t$, Butterfly Spread arbitrage is precluded if and only if the function $\gr:\RR \to \RR$ defined by 
\begin{equation}\label{eq:functiong} 
\gr(k):= \left(1 - \frac{k\partial_kw(k,t)}{2w(k,t)}\right)^2 - \frac{\partial^2_kw(k,t)}{4}\left(\frac{1}{w(k,t)}+\frac{1}{4}\right) + \frac{\partial_{kk}w(k,t)}{2},
\end{equation}
is a positive distribution,
with $\partial_{kk}w(\cdot, \cdot)$ defined in the distributional sense.
This condition in turn is equivalent to the Call price function being convex~\cite[Proposition~4.8]{JGMN16}. 
Assumption~\ref{assumption:totalVarConditions} ensures that $\partial_tw(k,t)$ is well defined 
for all $t>0$
and $\partial_kw(k,t)$ can be taken as right of left derivative at~$k$ if~$w$. 
Any valid extrapolation of the total implied variance for a fixed maturity must satisfy Roger Lee's conditions and 
be arbitrage-free. 
We start with the following simple result, proved in Appendix~\ref{app:lemLinearVarNoArb}:
\begin{lemma}\label{lem:LinearVarNoArb}
Fix a maturity $t>0$.
\begin{itemize}
\item (Right wing) 
For fixed constants $a_0, a_1\in\RR_+$ consider the function $w(k,t)\to a_1k+a_0$.
Then the function~$\gr$ is non-negative on $[k^*(a_0,a_1),\infty)$ if and only if $a_1\in[0,2]$, 
where $k^*(a_0,a_1)$ is a positive constant that depends on $a_0$ and $a_1$;
\item (Left wing) 
For fixed constants $a_0, a_1\in\RR_+$ consider the function $w(k,t)\to a_1|k|+a_0$.
Then the function~$\gr$ is non-negative on $[-\infty, k^*(a_0,a_1)]$ if and only if $a_1\in[0,2]$, 
where $k^*(a_0,a_1)$ is a negative constant that depends on $a_0$ and $a_1$.
\end{itemize}
\end{lemma}

\begin{assumption}\label{assumption:higherMoments}
There exist $p^*, q^*>0$ such that there is at least one market model $\mu\in\MMM_{\Cb}$ 
under which~$S$ admits moments of order at least $1+p^*$ 
and negative moments of order at most~$q^*$ up to~$T$. 
\end{assumption}
The set of martingale measures that satisfies Assumption~\ref{assumption:higherMoments} is defined as 
\begin{equation}\label{eq:martingaleHigherMoments}
\MMM^{p^*,q^*} := \MMM \cap  \left\{\mu\in\Pp(\Omega): 
\left\langle\omega^{1+p^*} + \omega^{-q^*}, \mu\right\rangle < \infty\right\},
\end{equation}
and the set of market models satisfying Assumption~\ref{assumption:higherMoments} is then defined as 
\begin{equation}\label{eq:modelHigherMoments}
\MMM^{p^*,q^*}_{\Cb} := 
\MMM_{\Cb}\cap\MMM^{p^*,q^*} .
\end{equation} 
Introduce the functions $f(x) := x^{1+p^*} + x^{-q^*}$ on~$\RR_+$
and $h : \Omega \to \RR$ as 
\begin{equation}\label{eq:functionHMomentsCh2}
h(\omega) := \sum_{t\in\Tt}f(S_t(\omega)).
\end{equation}

The following assumptions allow us to define a proper extrapolation of the total implied variance: 
\begin{assumption}[Left wing]\label{assumption:leftWing}
For any $t\in\Tt$, $K_1^t>0$, and the left wing is extrapolated as
\begin{equation}\label{eq:totalVarExtrapolationLeft}
w(k,t) := f_{L}(k - k^{t}_{1}, t) + w(k^{t}_{1},t), \qquad\text{for all }t\in\Tt, k<k_1^t:=\log(K_1^t),
\end{equation} 
where the function $f_L:\RR\times\Tt \to \RR_+$ satisfies
\begin{enumerate}[(A)]
	\item $f_L(0,\cdot) = 0$;
	\item $f_L(k,\cdot) = \Oo(\psi(q)|k|)$ for small enough~$k$ and $0<q<q^*$ 
	such that~$\gr$ is non-negative on $(-\infty,k^{t}_1)$;
	\item $\partial_tf_L(\cdot,t) \geq 0$, for any $t\in\Tt$.
\end{enumerate} 
\end{assumption}
\begin{assumption}[Right wing]\label{assumption:rightWing}
For $t\in\Tt$, $K^t_*=\infty$, and the right wing extrapolation reads
\begin{equation}\label{eq:totalVarExtrapolationRight}
	w(k,t) := f_{R}(k-k^{t}_{\kappa(t)},t) + w(k^{t}_{\kappa(t)},t),
	\qquad\text{for all }t\in\Tt, k > k^t_{\kappa(t)} := \log(K^t_{\kappa(t)})
\end{equation} 
where the function $f_R:\RR\times\Tt \to \RR_+$ satisfies
\begin{enumerate}[(A)]
	\item $f_R(0,\cdot) = 0$;
	\item $f_R(k,\cdot) = \Oo(\psi(p)k)$ for large enough~$k$ and $0<p<p^*$ 
	such that~$\gr$ is non-negative on $(k^{t}_{\kappa(t)}, \infty)$;
	\item $\partial_tf_R(\cdot,t) \geq 0$ for all $t\in\Tt$.
\end{enumerate} 
\end{assumption}
Here, the function $\psi:\RR\to [0,2]$ defined by $\psi(z) := 2-4(\sqrt{z(z+1)}-z)$
was introduced by \cite{Lee} and gives the precise slope of the total variance in the wings
as a function of the highest (absolute) moment of the underlying stock price.
Assumptions~\ref{assumption:leftWing} and~\ref{assumption:rightWing}
imply that extrapolation can be done linearly as long as the resulting total implied variance surface is consistent with the observed market prices (Assumptions~\ref{assumption:leftWing}(A) and~\ref{assumption:rightWing}(B))
and free of arbitrage, i.e. Assumptions~\ref{assumption:leftWing}(B)-(C) and~\ref{assumption:rightWing}(B)-(C) are satisfied.
In particular 
Assumptions~\ref{assumption:leftWing}(B) and~\ref{assumption:rightWing}(B) ensure the extrapolation is free of Butterfly Spread arbitrage and can be checked using results in Lemma~\ref{lem:LinearVarNoArb}.
Assumptions~\ref{assumption:leftWing}(C) and~\ref{assumption:rightWing}(C) ensure the extrapolation is free of Calendar Spread arbitrage.

As the underlying can be treated as an option with moneyness equal to zero, 
one can interpolate linearly between the traded option with the smallest available moneyness and 
the option with the zero moneyness. 
Therefore Assumption~\ref{assumption:leftWing} appears superfluous. 
However linear interpolation is only a crude approximation 
of the marginal distribution's behaviour near zero, 
whereas specifying extrapolation of the left wing of the smile
allows for a finer approximation (albeit parametric).

\begin{lemma}\label{lemma:WeakArbVarExtrapCh2}
Assume that the set of traded option prices $\Cb$ is free of strong model-independent arbitrage. Then, under Assumptions~\ref{assumption:leftWing} and~\ref{assumption:rightWing}, the Call price surface 
resulting from the total implied variance extrapolation is free of weak free lunch.
\end{lemma}
\begin{proof}
For each maturity $t\in\Tt$ define a probability measure $\mu_t$ on the state space $\Omega_t$ as
\begin{equation}\label{eq:BreedenLitzenbergerCh2}
\mu_{t}([0,K]) = 1 + \partial_+c(K,t)
\qquad\text{and}\qquad
\mu_{t}([K,+\infty)) = -\partial_-c(K,t).
\end{equation}
Assumption~\ref{assumption:rightWing} implies that
$\lim_{k\uparrow\infty}c_{\BS}(k,\sqrt{w(k,t)}) = 0$ and $\lim_{k\downarrow-\infty}c_{\BS}(k,\sqrt{w(k,t)})=1$
for each $t\in\Tt$ as a consequence of 
Assumption~\ref{assumption:leftWing} and therefore the expectation of~$S_t$ under~$\mu_t$ is equal to~$1$. 
The authors in~\cite{BL78} showed that the risk-neutral measure of normalised asset returns~$\mu_t$ at maturity $t\in\Tt$ can indeed be constructed that way. 
Moreover as the Call price surface resulting from the total implied variance extrapolation is free of Calendar spread arbitrage, then
\begin{equation}\label{eq:ConvexOrderSufficientConditionCh2}
\int_{0}^\infty(x-K)_+\mu_{t_1}(\rd x)\leq \int_{0}^\infty(x-K)_+\mu_{t_2}(\rd x),
\end{equation}
for any $t_{1}\leq t_{2}$, $K\in\RR_{+}$ which is sufficient for $\mu_{t_1}$ 
and $\mu_{t_2}$ to be in convex order~\cite{B12}.
Strassen's theorem~\cite{Strassen65} then yields the existence of a martingale measure $\mu\in\MMM$ 
with marginals~$\mu_t$ for all $t\in\Tt$. 
Hence there exists at least one market model consistent with traded Call options prices and the extrapolation of the total variance and absence of weak free lunch follow by Theorem~\ref{thm:FTAPnoOrderUnit}.
\end{proof}

The extrapolation of the total implied variance restricts the feasible set 
of the dual problem~\ref{eq:dual}.
To avoid emergence of a duality gap, 
the feasible sets of the primal problem~\ref{eq:primal} must be enlarged:
untraded Call options priced from the extrapolation must be added to the set of static positions~$\Ss$, 
and the set becomes infinite-dimensional. 
Addition of untraded options does not create a duality gap as the resulting set of traded options is free of weak free lunch and the duality results above still apply.

\subsection{Weak arbitrage and duality}\label{sec:CaseOfFinitelyManyOptionsCh2}
As above, absence of week free lunch implies that a market model $\mu\in\MMM_{\Cb}$ 
corresponds to a strictly positive linear functional and hence $\langle m, \mu\rangle > 0$ for all $m\in \Mf_{++}$.
This is a rather strict assumption as it is possible to have Butterfly Spreads traded on the market at zero price 
and find a corresponding market model as shown in~\cite[Theorems 3.1, 4.2]{DH07}.
Moreover, as we shall explore semi-infinite approximations to the primal~\ref{eq:primal} and dual~\ref{eq:dual} problems in the next section~\ref{sec:semiInfiniteLP}, when only finitely many options are available for each maturity $t\in\Tt$ the notion of weak free lunch does not appear to be helpful. 
We thus introduce a notion of weak arbitrage as in~\cite[Definition 2.3]{CoxObloj11}.
\begin{definition}\label{defn:WeakArb}
The pricing functional~$\rho$ in~\eqref{eq:rhoFunctional} admits weak arbitrage on~$\Mf$ 
if for any model $\mu\in\MMM$, there exists $m\in \Mf$ such that 
$\rho(m)\leq 0$, 
but $\mu(\{\omega \in \Omega : m(\omega)\geq 0\}) = 1$ and $\mu(\{\omega \in \Omega : m(\omega)> 0\}) > 0$.
\end{definition}
Under weak arbitrage, $\MMM_{\Cb}$ is empty.
Clearly, strong model-independent arbitrage opportunities are also weak arbitrage opportunities.
This definition of weak arbitrage allows the use of the result~\cite[Theorem~4.2]{DH07} stating that when only finitely many options are traded on the market, absence of weak arbitrage is equivalent to existence of a market model.
It is easily seen that absence of weak arbitrage implies that if there exists a claim $m\in \Mf_+$ 
with market price $\rho(m) = 0$ then
$\mu(\{\omega\in \Omega : m(\omega) > 0 \}) = 0$
for any market model $\mu\in\MMM_{\Cb}$.
With $\Fr_0 := \{m\in \Mf : \rho(m) = 0\}$ denoting 
the set of all traded claims available on the market at price zero, we introduce the convex cone 
\begin{equation}\label{eq:coneW}
\Wf := \Fr_0 \cap (\Cc_{h})_+(\Omega).
\end{equation}
This cone highlights a fundamental issue in strong model-independent arbitrage:
assume that this cone is generated by finitely many traded Butterfly Spreads traded at zero price for each $t\in\Tt$.
For fixed $t\in\Tt$ and any three strikes $K^t_{i-1}<K^t_i<K^t_{i+1}$ 
(with $1<i<\kappa(t)$) the payoff of a Butterfly Spread is 
$$
\alpha(S_t - K^t_{i-1})_+ -(\alpha + \beta)(S_t - K^t_{i})_+ + \beta(S_t - K^t_{i+1})_+,
$$
where $\alpha := 1/(K^t_{i}-K^t_{i-1})$ and $\beta := 1/(K^t_{i+1}-K^t_{i})$.
If it is traded at zero price, 
then absence of weak arbitrage implies that any market model $\mu\in\MMM_{\Cb}$ 
places no mass on the open interval $(K^t_{i-1}, K^t_{i+1})$. 
Clearly the collection of such open sets is closed under taking finite intersections and unions.
Basically, any market model consistent with butterflies priced at zero 
places no mass on the open interval where the payoff of a butterfly is strictly positive.
In that case, there is strong model-independent arbitrage and one cannot use strictly positive linear functionals and extensions thereof, 
but rather just positive functionals, 
which also implies that the ordering on the space of claims needs to be amended.
Let us introduce such an ordering on~$\Cc_{h}(\Omega)$ 
by defining a `trans-positive' closed convex cone
\begin{equation}\label{eq:TransPositiveCone}
\Jj := \overline{(\Cc_{h})_+(\Omega) - \Wf},
\end{equation}
where the closure is taken with respect to the norm topology on $\Cc_{h}(\Omega)$.
This set was introduced by \cite{Clark06} in order to provide an infinite-dimensional 
generalisation of the classical Farkas condition regarding the feasibility of finite-dimensional linear programmes.
Since $0\in \Jj$, we can introduce a new ordering ``$\succeq$'' 
on $\Cc_{h}(\Omega)$ 
such that the relation $f_1\succeq f_2$ holds if and only if $f_1 - f_2 \in \Jj$.
The following lemma shows how 
the negative polar $\Jj^*\subset \Pp_h(\Omega)$ (Definition~\ref{defn:NegativePolar})
characterises weak arbitrage.
\begin{lemma}
Absence of weak arbitrage implies that $\MMM^{p^*,q^*}_{\Cb}\subset\Jj^*$.
\end{lemma}
\begin{proof}
For any $\mu \in \MMM^{p^*,q^*}_{\Cb}$,
the inequality $\langle f, \mu\rangle \geq 0$ holds for all $f\in (\Cc_{h})_+(\Omega)$,
and for any $w\in \Wf$, $\langle w, \mu\rangle$ is null
by absence of weak arbitrage. 
So for any $f\in(\Cc_{h})_+(\Omega)$ and $w\in \Wf$ we have 
$0 \leq \langle f,\mu\rangle = \langle f, \mu\rangle - \langle w, \mu\rangle
 =  \langle f - w, \mu\rangle$.
Since $f - w \in \Jj$, the lemma follows by definition of the negative polar~$\Jj^*$.
\end{proof}
The above analysis also remains the same for any $j$ on the boundary of $\Jj$. 
In particular let $j:=\lim_{n\uparrow\infty}j_n = \lim_{n\uparrow\infty}(f_n-w_n)$ 
and by linearity of the inner product for any $\mu \in \MMM_{\Cb}$ we have 
\begin{equation}
\left\langle j,\mu\right\rangle = \left\langle\lim_{n\uparrow\infty}(f_n-w_n),\mu\right\rangle = 
\lim_{n\uparrow\infty}\left\langle(f_n-w_n),\mu\right\rangle = 
\left\langle\lim_{n\uparrow\infty} f_n,\mu\right\rangle = 
\left\langle f,\mu\right\rangle, 
\end{equation}
where $f\in (\Cc_h)_+(\Omega)$ as the positive cone is closed in the weak topology.

For an option with payoff $\Phi\in \Cc_{l}(\Omega)$ 
(note that as we consider the case when finitely many options are available for hedging, only options with payoffs that grow at most linearly can be superhedged. Of course, if the state space is restricted to be a compact subset of $\Omega$, superlinear payoffs can only be considered),
define now the super-hedging problem
\begin{equation}\label{eq:primalWeakArbChapter2}
{}^*\val_p(\Phi) := \inf\left\{\barho(m) : m\in\overline{\Mf}\text{, }m-\Phi \in \Jj\right\},
\end{equation}
and its associated dual
\begin{equation}\label{eq:dualWeakArbChapter2}
{}^*\val_d(\Phi) := \sup\left\{\langle \Phi, \mu\rangle: \mu\in\MMM^{p^*,q^*}_{\Cb}\right\}. 
\end{equation}
Symmetrically, the sub-hedging primal problem is defined as ${}_*\val_p(\Phi) = -{}^*\val_p(-\Phi)$
and the sub-hedging dual problem as ${}_*\val_d(\Phi) = -{}_*\val_d(-\Phi)$.
To state the required duality, we impose the following assumption:
\begin{assumption}\label{assumption:technicalassumption}
If there exists a continuous linear extension $\pi : C_h(\Omega) \to \RR$ of $\rho$, 
then for all $(f_n)_{n\in\NN}\in C_h(\Omega)$ decreasing pointwise to zero,
$\lim_{n\uparrow\infty} \pi(f_n) = 0$.
\end{assumption}

\begin{theorem}\label{thm:dualityWeakArbChapter2}
Absence of weak arbitrage implies no duality gap between~\eqref{eq:primalWeakArbChapter2} 
and~\eqref{eq:dualWeakArbChapter2} on~$\Cc_{l}(\Omega)$,
and likewise for the sub-hedging problems.
\end{theorem}
\begin{proof}
We only prove the super-hedging case as the sub-hedging one follows by symmetry,
and we follow closely the arguments from Theorem~\ref{thm:superReplication}.
We assume that $\Phi\notin\overline{\Mf}$, otherwise the theorem is trivial.
Absence of weak arbitrage implies there exists a market model $\mu_0\in\MMM_{\Cb}$ with
$\EE^{\mu_0}\{\Phi\} := \langle \Phi, \mu\rangle \leq {}^*\val_p(\Phi)$ and
fix $\lambda\in(\EE^{\mu_0}\{\Phi\} ,{}^*\val_p(\Phi))$.
Let $G := \Span\{\overline{\Mf}, \Phi \}$ and define $\eta : G \to \RR$ as 
$\eta(g) := \eta(m + t\Phi) = \barho(m) + t\lambda$. 
We now show that $\eta$ is positive on $\Jj_G:=\Jj\cap G$. 
Let $g=m+t\Phi \in \Jj_G$ and consider three cases.
If $t=0$ then $\eta(g) = \barho(m) \geq 0$. 
If $t<0$ then $(-t)^{-1}m\succeq \Phi$ and $(-t)^{-1}\barho(m) \geq {}^*\val_p(\Phi) > \lambda$.
Similarly if $t>0$ then $\Phi\succeq (-t)^{-1}m$ and hence $\barho(m) > -t\lambda$.
It also follows that if $t\neq 0$ then $\eta(g) > 0$.

As $\eta$ is linear and dominated by a convex function ${}^*\val_p$ 
(as the function $l$ defined in~\eqref{eq:functionHLinear} is an element of~$\Mf$, 
the function $-\infty<{}^*\val_p(f)<\infty$ for all $f\in \Cc_{l}(\Omega)$)  
hence by Hahn-Banach Extension Theorem there exists a linear extension of~$\pi$ 
to the whole space~$\Cc_{l}(\Omega)$ such that~$\pi$ is dominated by~${}^*\val_p$.
For $j\in \Jj$ we have $0\succeq -j$ and $\pi(-j) \leq {}^*\val_p(-j) \leq \barho(0) = 0$
thus $\pi(j) \geq 0$ by linearity of~$\pi$. 
As $0\in W$ it implies that $\pi$ is a positive linear functional and as $\Cc_{l}(\Omega)$ is a Banach lattice 
it follows by~\cite[Theorem 1.36]{AliTourky} that~$\pi$ is continuous and by Assumption~\ref{assumption:technicalassumption} it can be represented as a Borel probability measure, 
i.e. $\pi\in\Pp_l(\Omega)$. 
Moreover~$\pi$ also extends~$\rho$ and hence gives a market model. 

By construction $\pi(\Phi) = \eta(\Phi) = \lambda$. 
Since~$\pi$ is a market model, it is a feasible solution to~\eqref{eq:dualWeakArbChapter2} 
and $\lambda = \pi(\Phi) \leq {}^*\val_d(\Phi)$. 
As $\lambda\in(\EE^{\mu_0}\{\Phi\} ,{}^*\val_p(\Phi))$ was chosen arbitrarily, 
hence ${}^*\val_d(\Phi) = {}^*\val_p(\Phi)$.
\end{proof}
The primal~\eqref{eq:primalWeakArbChapter2} and the dual~\eqref{eq:dualWeakArbChapter2} problems 
can be extended to the case when $\Phi\in \Uu_{l}(\Omega)$
by defining the extension to the primal problem $\overline{\val}_p : \Uu_{l}(\Omega) \to \overline{\RR}$,
with $\overline{\RR}:= [-\infty,+\infty]$, as
\begin{equation}\label{eq:primalUpperSemiContWeakArbCh2}
\overline{\val}_p(\Phi) := \inf\left\{{}^*\val_p(f) : f\in \Cc_{l}(\Omega), 
f\geq \Phi\text{ on } \Omega \right\}.
\end{equation}
The corresponding extension to the dual problem $\overline{\val}_d : \Uu_{l}(\Omega) \to \overline{\RR}$ 
is defined as 
\begin{equation}\label{eq:dualUpperSemiContWeakArbCh2}
\overline{\val}_d(\Phi) := \sup\left\{\langle\Phi, \mu\rangle: \mu\in\MMM_{\Cb} \right\}.
\end{equation}
The sub-hedging primal problem can be extended to $\Phi\in \Ll_l(\Omega)$ in a similar way. 

If the convex cone~$\Wf$ in~\eqref{eq:coneW} is trivial, i.e. $\Wf = \{0\}$, 
then the trans-positive cone is reduced to the positive cone $(\Cc_{l})_+(\Omega)$, i.e. 
$\Jj = \overline{(\Cc_{l})_+(\Omega) - \Wf} = \overline{(\Cc_{l})_+(\Omega)} = (\Cc_{l})_+(\Omega)$.
Then the definitions of the primal~\eqref{eq:primalUpperSemiContWeakArbCh2} 
and the dual~\eqref{eq:dualUpperSemiContWeakArbCh2} coincide with the definitions of the primal~\eqref{eq:primal} and the dual~\eqref{eq:dual} programmes. 
In particular the super-hedging primal problem for any $\Phi \in \Cc_{l}(\Omega)$ is written as 
\begin{equation}
{}^*\val_p(\Phi) := \inf\left\{\barho(m) : m\in\overline{\Mf}\text{, }m-\Phi \in (\Cc_{l})_+(\Omega)\right\},
\end{equation}
and coincides with $\overline{\val}_p(\Phi)$.
The sub-hedging problems are likewise reduced to~\eqref{eq:subhedgePrimal} and~\eqref{eq:subhedgeDual}.

\section{Reduction to the semi-infinite case}\label{sec:semiInfiniteLP}
The literature on computational methods for sub- and super-hedging problems has been rather sparse,
with the recent exceptions~\cite{Alfonsi, Benamou, GuObloj}. 
Guo and Ob{\l}{\'o}j~\cite{GuObloj} devtelop computational methods to solve the martingale optimal transport (MOT) problem via discretisation and optimisation techniques. 
In particular, they consider an approximation of the MOT via a series of linear programmes. 
To do so, discretisation of the marginal distributions is introduced along with approximation of the martingale condition on a finite number of constraints. 
They introduce the notion of $\eps$-approximating martingale measures, 
and obtain an upper bound on the speed of convergence in the one-dimensional case. 
Assuming existence of moments of the marginal distribution, 
the numerical implementation relies on computing the Wasserstein distance between
the marginal distribution and its approximation. 
They propose two generic approaches to solve this, one in case where the density function of the marginal distribution is known and the second one where one can sample from the marginal. 

We discuss here a reduction of the infinite-dimensional problems~\eqref{eq:primal}-\eqref{eq:dual} 
to the semi-infinite case, with a view towards numerical implementation. 
We first select a finite subset of traded options approximating the set of static positions~$\Ss$ 
from Definition~\ref{defn:staticPosition}.
When only finitely many Call options are traded, we perform extrapolation of the total implied variance according to Assumptions~\ref{assumption:leftWing} and~\ref{assumption:rightWing}, 
and include Call options with prices corresponding to such extrapolation. 
Note that those options may not be traded on the market.
We define a vector of Call option payoffs as 
\begin{equation}\label{eq:OptionPayOffsVectorCh3}
\CC := \left((S_t - K^t_1)_+,\ldots,(S_t - K^t_{\kappa(t)})_+\right)_{t\in\Tt}\in\RR^{\df},
\end{equation}
where $\df:=\sum_{t\in\Tt}\kappa(t)<\infty$, 
and the vector of corresponding market prices as before as
\begin{equation}\label{eq:OptionPriceVectorCh3}
\Cb:= (c(K^t_1,t),\ldots,c(K^t_{\kappa(t)},t))_{t\in\Tt}\in\RR^{\df}.
\end{equation}
We shall also write  
$\CC(\omega) := \left((S_t(\omega) - K^t_1)_+,\ldots,(S_t(\omega) - K^t_{\kappa(t)})_+\right)_{t\in\Tt}$
to denote the evaluation of the Call options payoffs at $\omega\in\Omega$.
\begin{assumption}\label{assumption:noWeakArbChapter3}
The prices~$\Cb$ preclude weak arbitrage 
and~$\Wf$ in~\eqref{eq:coneW} is trivial, 
i.e. $\Wf = \{0\}$.
\end{assumption}
As mentioned previously, when $\Wf = \{0\}$,
the super- and sub-hedging problems~\eqref{eq:primalUpperSemiContWeakArbCh2} and~\eqref{eq:dualUpperSemiContWeakArbCh2} are equivalent to~\eqref{eq:primal} and~\eqref{eq:subhedgePrimal} respectively. 
The set of approximate static positions is now $\widetilde{\Ss} := \RR\times\Span\{\CC\}$,
the first component representing the cash position.
We also discretise the set of trading strategies $\Hh = \RR\times\prod_{j=1}^{n-1}\Cc_{b}(\RR^j_+)$
from Definition~\ref{def:tradingStrategy}. 
For a rational number $\alpha\in\mathbb{Q}$ let $K^j_{\alpha} : = [0,\alpha]^j$ where $j=1,\ldots,n-1$ and define a set of functions 
$B:=\{\theta^{t_j}_{i} \in C_b(\RR^j_+)\text{, }j=1,\ldots,n-1\text{, }i\in\NN\}$ such that for each $j$ and $\alpha$ the set $\{\ind_{K^j_{\alpha}}\theta^{t_j}_{i}\text{, }i\in\NN\}$ is dense in $C(K^j_{\alpha})$. 
Let us also define a finite subset $B_j:= \{\theta^{t_j}_1,\ldots,\theta^{t_j}_{d(t_j)}\}$ with $d(t_j)<\infty$ of elements in $B$ for each $j=1,\ldots,n-1$
(for instance, one can take a set of monomials defined on $K^j_{\alpha}$ for each $j$ and $\alpha$ and extend each element in the set to $\RR_+^j$ such that the extension is equal to the maximum of the element on $K^j_{\alpha}$ on the complement of $K^j_{\alpha}$ and is equal to the element itself otherwise).
Then a discretised trading strategy $\Theta\in\widetilde{\Hh}:= \RR \times\prod_{j=1}^{n-1}\Span\{B_j\}$ is defined as follows and an element $\Theta\in\widetilde{\Hh}$ reads
$$
\Theta(\omega) := 
\left(a_0, \langle\aaa^{1}, \theta^{1}(\omega)\rangle,\ldots, \langle\aaa^{n-1}\theta^{n-1}(\omega)\rangle\right),
$$
for each $\omega\in\Omega$, where $a_0\in\RR$, $\aaa^{j}\in\RR^{d(t_j)}$, and 
$\theta^{j}(\omega)\in\RR^{d(t_j)}_+$ 
are the evaluation vectors of basis functions for each time period~$t_j$. 
Note that $\theta^{j}(\omega) := \theta^{j}(\Pr\omega)$, 
where $\Pr\omega$ is the projection of $\omega\in\Omega$ onto $\RR^j_+$.
Note that we use the same notation $\langle\cdot,\cdot\rangle$ to denote the Euclidean inner product, but this should hopefully not create any confusion.
The payoff of a discretised trading strategy $\Theta\in\widetilde{\Hh}$ then reads
\begin{equation}\label{eq:discretisedHedgeChapter3}
\left(\Theta\bullet S\right)_T = a_0(S_{t_1} - s_0) + \sum_{j=1}^{n-1}\sum_{i=1}^{d(t_j)}a^{j}_i\theta^{j}_i\left(S_{t_{j+1}} - S_{t_j}\right).
\end{equation}
The initial cost~\eqref{eq:NotPi0T} of a discretised hedging portfolio 
$(\widetilde{\ff},\Theta)\in\widetilde{\Ss}\times\widetilde{\Hh}$ now reads
$\Pi_{t_0}(\widetilde{\ff},\Theta) = \langle \Cb, \ww\rangle + \lambda$,
where $\lambda\in\RR$, the vector $\ww = (w^t_1,\ldots,w^t_{\kappa(t)})_{t\in\Tt}\in\RR^{\df}$ with entries denoting portfolio weights in available options and $\langle \cdot,\cdot\rangle$ is the inner product in~$\RR^{\df}$. 
We also write the payoff of the hedging portfolio $(\widetilde{\ff},\Theta)$ at the final maturity,
$\Pi_T(\widetilde{\ff},\Theta) = A_{\lambda}^{\Theta}(\ww)$,
where the linear map~$A$ is defined as 
\begin{equation}\label{eq:Adefn}
A_{\lambda}^{\Theta}(\ww)  := \lambda + \sum_{t\in\Tt}\sum_{i=1}^{\kappa(t)}w^t_i(S_t - K^t_i)_+ + (\Theta \bullet S)_T = \lambda + \CC \ww + (\Theta\bullet S)_T.
\end{equation}
We can then write a problem of super-hedging an option with the upper semi-continuous payoff $\Phi \in \Uu_{l}(\Omega)$ bounded above by a linear function~$l$ defined in~\eqref{eq:functionHLinear} as 
\begin{equation}\label{eq:primalSemiInfinite}
\overline{\val}_p(\Phi)
 := \inf\left\{ \lambda + \left\langle \Cb, \ww\right\rangle : (\ww,\lambda, \Theta)\in\overline{\Ff}_p
\right\}.
\end{equation}
Even though this definition, because of the discretisation, is different from
its infinite-dimensional counterpart~\eqref{eq:primalUpperSemiContWeakArbCh2}, 
we keep the same notation without confusion.
The feasible set~$\overline{\Ff}_p$ is defined as 
\begin{equation}\label{eq:primalPerturbedFeasibleSet}
\overline{\Ff}_p := \left\{(\ww,\lambda,\Theta)\in\RR^{\df+1}\times\widetilde{\Hh} :
A_{\lambda}^{\Theta}(\ww;\omega) - \Phi(\omega)\geq 0\text{ for all }\omega\in\Omega\right\},
\end{equation}
and the associated dual problem has the form
\begin{equation}\label{eq:dualSemiInfinite}
\overline{\val}_d(\Phi) := \sup\left\{ \langle \Phi, \mu\rangle: \mu\in\widetilde{\MMM}^{p^*,q^*}_{\Cb}\right\},
\end{equation}
where the set of Borel probability measures that re-price the discretised portfolios in $\widetilde{\Ss}\times\widetilde{\Hh}$ reads
$$
\widetilde{\MMM}^{p^*,q^*}_{\Cb}
 := \left\{ \mu\in\Pp_h(\Omega): 
 \langle \Pi_T(\widetilde{\ff},\Theta;\cdot), \mu\rangle
= \Pi_0(\widetilde{\ff},\Theta)\text{, }(\widetilde{\ff},\Theta)\in\widetilde{\Ss}\times\widetilde{\Hh}\right\},
$$
with the function~$h$ defined in~\eqref{eq:functionHMomentsCh2}, 
and the real constants $p^*, q^*>0$ in Assumption~\ref{assumption:higherMoments}.
The sub-hedging primal and dual problems can be defined in a similar manner.
We now show that the primal 
and their corresponding dual problems in the sub- and super-hedging cases admit no duality gap.
\begin{proposition}\label{propostion:semiinfduality}
Under Assumptions~\ref{assumption:higherMoments} and~\ref{assumption:noWeakArbChapter3},
there is no duality gap  between~\eqref{eq:primalSemiInfinite} and~\eqref{eq:dualSemiInfinite}.
\end{proposition}
\begin{proof}
By Lemma~\ref{lemma:WeakArbVarExtrapCh2}, Assumptions~\ref{assumption:higherMoments} and~\ref{assumption:noWeakArbChapter3} imply absence of weak free lunch. 
Moreover as the riskless bond satisfies Assumption~\ref{assumption:risklessBond}, 
the statement follows from Theorem~\ref{thm:superReplication}.
\end{proof}

\begin{remark}
As the sub-hedging primal problem can be represented as $\underline{\val}_p(\Phi) = - \overline{\val}_p(-\Phi)$
and the sub-hedging dual problem
is represented in terms of the super-hedging dual problem~\eqref{eq:dualSemiInfinite} as  
$\underline{\val}_d(\Phi) = - \overline{\val}_d(-\Phi)$, Proposition~\ref{propostion:semiinfduality} can be applied to sub-hedging problems as well.
\end{remark}

This discretisation setting is justified by the following result, proved in Appendix~\ref{app:thmConvergence}, which shows that when the number of elements in the basis of the set of discretised 
trading strategies~$\widetilde{\Hh}$ increases to infinity, 
the semi-infinite primal~\eqref{eq:primalSemiInfinite}
and the dual~\eqref{eq:dualSemiInfinite} problems converge to the values of the infinite-dimensional problems defined in~\eqref{eq:primal} and~\eqref{eq:dual} respectively.

\begin{theorem}\label{thm:Convergence}
Under Assumptions~\ref{assumption:higherMoments} and~\ref{assumption:noWeakArbChapter3},
as~$r := \min_{t\in\Tt}\{d(t)\}$ tends to infinity, 
the values of both semi-infinite programmes converge to the values of their infinite-dimensional counterparts.
\end{theorem}
The form of the discretisation provides information about the convergence:
as the latter is refined, the feasible set~\eqref{eq:primalPerturbedFeasibleSet} for the super-hedging problem becomes larger, and the infimum~\eqref{eq:primalSemiInfinite} decreases.

\section{Perturbation analysis of model-independent hedging problems}\label{sec:PerturbationAnalysis}
Extrapolation of the total implied variance in Section~\ref{sec:extrapVarCh2}
restricts the feasible sets of the dual problem~\eqref{eq:dualUpperSemiContWeakArbCh2}
as well as the feasible set of its semi-infinite approximation~\eqref{eq:dualSemiInfinite}.
On the other hand the feasible sets of the primal problem~\eqref{eq:primalSemiInfinite} is enlarged by adding non-traded Call options with prices consistent with extrapolation.
As this assumption is exogenous, we study now the sensitivity of the optimal values 
of the dual problems to extrapolation of the total implied variance. 
We embed the semi-infinite approximations to the primal and dual problems 
into a family of perturbed problems, where the perturbations are changes in input Call option prices,
and use the language of directional derivatives to provide a rigorous sensitivity analysis.
 
\subsection{Perturbation analysis}\label{sec:perturbationAnalysis}
We embed the primal~\eqref{eq:primalSemiInfinite} and dual~\eqref{eq:dualSemiInfinite} problems 
into a family of perturbed problems by introducing a vector 
$\uu := (u^t_1,\ldots,u^t_{\kappa(t)})_{t\in\Tt}\in\RR^{\df}$ of price perturbations. 
Given an option with payoff $\Phi\in \Uu_h(\Omega)$,
let $\widetilde{\val}_p: \RR^{\df} \to \overline{\RR}$ denote the value of
the perturbed super-hedging primal problem
\begin{equation}\label{eq:primalPerturbed}
\widetilde{\val}_p(\uu) := \inf\left\{\lambda + \left\langle \Cb+\uu, \ww\right\rangle : 
(\ww,\lambda,\Theta)\in \overline{\Ff}_p
\right\},
\end{equation}
where $\overline{\Ff}_p$ is the feasible set defined in~\eqref{eq:primalPerturbedFeasibleSet}.
The explicit dependence on the payoff~$\Phi$ in the notations is dropped for simplicity,
since our aim here is to focus more on the perturbation~$\uu$ of the initial input, 
rather than on the final payoff.
The value function $\widetilde{\val}_p$ is convex 
and $\widetilde{\val}_p(0)$ coincides with the value of the unperturbed primal problem~\eqref{eq:primalSemiInfinite}. 
Defining the Lagrangian function  
\begin{equation}\label{eq:LagrangianFunction}
L^\Theta_\lambda(\ww,\mu)
 := \lambda + \left\langle \Cb, \ww\right\rangle - \left\langle A_{\lambda}^{\Theta}(\ww)-\Phi,\mu\right\rangle,
\end{equation}
we can then write, by definition of~$\overline{\Ff}_p$,
\begin{equation}\label{eq:lagrangiansup}
\sup_{\mu\in(\Mm_h)_+(\Omega)}\left\{L^\Theta_\lambda(\ww,\mu) + \langle\uu,\ww\rangle \right\}= 
\begin{cases}
\lambda + \left\langle \Cb+\uu, \ww\right\rangle, & \text{if } (\ww,\lambda,\Theta)\in\overline{\Ff}_p, \\
+\infty, & \text{otherwise},
\end{cases}
\end{equation}
which yields the equivalent formulation of the primal problem:
\begin{equation}
\inf_{(\ww,\lambda,\Theta)\in\RR^{\df+1}\times\widetilde{\Hh}}\sup_{\mu\in(\Mm_h)_+(\Omega)}\{L^\Theta_\lambda(\ww,\mu)
 + \langle\uu,\ww\rangle\}.
\end{equation}
On the other hand if the infimum is taken over $(\ww,\lambda,\Theta)\in\RR^{\df+1}\times\widetilde{\Hh}$ first, we obtain $$
\inf_{(\ww,\lambda,\Theta)\in\RR^{\df+1}\times\widetilde{\Hh}}\left\{L^\Theta_\lambda(\ww,\mu)
 + \langle\uu,\ww\rangle \right\} 
= \inf_{(\ww,\lambda,\Theta)\in\RR^{\df+1}\times\widetilde{\Hh}} \left\{ \left\langle \Phi ,\mu\right\rangle 
+\lambda + \left\langle \Cb+\uu, \ww\right\rangle - \left\langle A_{\lambda}^{\Theta}(\ww),\mu\right\rangle \right\}.
$$
The expression on the right is not equal to $-\infty$ if
$\lambda + \left\langle \Cb+\uu, \ww\right\rangle = \left\langle A_{\lambda}^{\Theta}(\ww),\mu\right\rangle$
for all $(\ww,\lambda,\Theta)\in\RR^{\df+1}\times\widetilde{\Hh}$.
Expanding the right-hand side according to Definition~\eqref{eq:Adefn} 
and comparing the terms on the left and the right of the equality we see that it holds if 
$$
\left\langle \lambda, \mu\right\rangle = \lambda,
\quad
\left\langle (\Theta \bullet S)_T,\mu\right\rangle = 0
\quad
\text{and}
\quad
\left\langle \CC \ww,\mu\right\rangle = \left\langle \Cb + \uu, \ww\right\rangle.
$$
In particular the last equality can be re-written as 
$$
0 = \left\langle \CC \ww,\mu\right\rangle - \left\langle \Cb + \uu, \ww\right\rangle = 
\left\langle \ww, \CC^*\mu\right\rangle - \left\langle \Cb + \uu, \ww\right\rangle = 
\left\langle \CC^*\mu - \Cb - \uu, \ww \right\rangle,
$$
where $\CC^*\mu$ defines the adjoint map of $\CC : \ww\mapsto \CC \ww \in \Cc_{h}(\Omega)$.
Since the inner product on the right-hand side is null for all $\ww\in\RR^{\df}$, then
$\CC^*\mu = \Cb + \uu$.
The perturbed dual problem thus reads
\begin{equation}\label{eq:dualPerturbed}
\widetilde{\val}_d(\uu) := \sup\left\{ \langle \Phi,\mu\rangle : \mu\in\MMM_\uu\right\},
\end{equation}
where $\MMM_\uu$ is the feasible set of all non-negative Borel measures that integrate
$h$ to a finite constant
\begin{equation}\label{eq:dualPerturbedFeasibleSet}
\MMM_\uu := \left\{\mu\in(\Mm_h)_+(\Omega) :
\left\langle(\Theta \bullet S)_T,\mu\right\rangle = 0 
\text{, }\CC^*\mu = \Cb + \uu\right\}
\end{equation}
satisfying the martingale condition for all $\Theta\in\widetilde{\Hh}$ 
and which are consistent with the perturbed Call prices.
The value $\widetilde{\val}_d(0)$ corresponds 
to that of the unperturbed dual problem~\eqref{eq:dualSemiInfinite}.

We now show that weak arbitrage prevents duality gap:
\begin{theorem}
Suppose that for some perturbation $\uu\in\RR^{\df}$, 
the prices $\uu + \Cb$ satisfy Assumption~\ref{assumption:noWeakArbChapter3}.
Then there is no duality gap between~\eqref{eq:primalPerturbed} and~\eqref{eq:dualPerturbed}.
\end{theorem}
\begin{proof}
Our proof relies on~\cite[Theorem~5.99]{BonnansShapiro00}, 
which characterises absence of duality gap as a condition on the range of the adjoint map~$\CC^*$, 
defined as the moment cone
\begin{equation}\label{eq:momentCone}
\Mbf :=
\left\{\uu\in\RR^{\df} : \text{ there exists }\mu\in(\Mm_h)_+(\Omega)\text{, } \uu = \CC^*\mu - \Cb, 
\left\langle(\Theta\bullet S)_T,\mu\right\rangle = 0
\text{ for all }\Theta\in\widetilde{\Hh} \right\}.
\end{equation}
If $\uu\in\interior \Mbf$, then there is no duality gap between the primal~\eqref{eq:primalPerturbed} 
and the dual~\eqref{eq:dualPerturbed} super-hedging problems.
Absence of weak arbitrage is equivalent~\cite[Theorem~4.2]{DH07} 
to the existence of a model $\mu\in\MMM_{\uu}$ for prices $\Cb + \uu$. 
Moreover following~\cite[Proof of Proposition~3.1]{DavisOblojRaval14},
in order to show $\uu\in\interior \Mbf$, it is sufficient to note that
for any entry $c(K^t_i,t)+u^t_i$ of the vector $\Cb + \uu$, 
the inequalities $(1-K^t_i)_+ < c(K^t_i,t)+u^t_i < 1$
hold for all $i=1,\ldots,\kappa(t)$ and $t\in\Tt$ as perturbed prices satisfy Assumption~\ref{assumption:noWeakArbChapter3}. 
As $\mu \mapsto \CC^*\mu$ is a continuous function on $\Pp_h(\Omega)$ by~\cite[Lemma 2.2]{BHP13} one can also find a real positive constant $\varepsilon>0$ such that any vector $\vv$
in the open ball $\mathcal{B}_{\varepsilon}(\Cb + \uu)$ centred around $\Cb + \uu$ 
satisfies Assumption~\ref{assumption:noWeakArbChapter3}, 
and therefore $\uu\in\interior \Mbf$ and the theorem follows.
\end{proof}
The condition on the moment cone in the proof goes back to~\cite[Chapter XII, Theorem~2.1]{KS66} in the context of generalised Tchebycheff inequalities, and can also be found in~\cite[Theorem~4.4]{AN87}.
A similar result was used in~\cite{DavisOblojRaval14} to prove absence of duality gap 
under absence of weak arbitrage opportunities. 
Having established absence of duality gap between the primal~\eqref{eq:primalPerturbed} 
and the dual~\eqref{eq:dualPerturbed}, 
we now discuss sensitivity of the programmes to the perturbation. 
In particular, the dual is continuous at~$\uu$;
moreover if the primal is finite at~$\uu$ we have the following:
\begin{proposition}\label{proposition:Sensitivity}
Assume there is no duality gap between the primal and the dual problems for some $\uu\in\RR^{\df}$. 
If the value of the primal at~$\uu$ is finite, 
then the dual is Hadamard directionally differentiable at~$\uu$, 
and the derivative in any direction $\hh\in\RR^{\df}$ reads
$$
(\widetilde{\val}_d)'(\uu,\hh)
 = \inf\left\{\left\langle \ww, \hh\right\rangle: \ww\in\widetilde{\mathfrak{S}}_{\uu}\right\}
\qquad\text{and}\qquad
(\utilde{\val}_d)'(\uu,\hh)
 = \sup\left\{\left\langle \ww, \hh\right\rangle: \ww\in\utilde{\mathfrak{S}}_{\uu}\right\},
$$
where 
$\widetilde{\mathfrak{S}}_{\uu}, \utilde{\mathfrak{S}}_{\uu}\subset \RR^{\df+1}\times\widetilde{\Hh}$ 
denote the set of optimal solutions of the primal problem at~$\uu$ in the super- and sub-hedging problems. 
\end{proposition}
\begin{proof}
We only prove the super-hedging case, as the sub-hedging one is analogous.
By a change of variables $\mu \mapsto -\mu$ we turn the dual problem into the minimisation problem 
\begin{equation}\label{eq:dualmin}
\varpi(\uu) := \inf\left\{\left\langle \Phi, \mu\right\rangle : -\mu\in\MMM_\uu\right\},
\end{equation}
and of course $\varpi(\uu) = - \widetilde{\val}_d(\uu)$. 
Let us now calculate the convex conjugate of~$\varpi$ at $\uu^*\in\RR^{\df}$
\begin{align*}
\varpi^*(\uu^*) &= \sup_{\uu\in\RR^{\df}} \left\{ \langle \uu, \uu^*\rangle - \varpi(\uu)\right\}
= 
\sup_{\mu\in(\Mm_h)_+(\Omega)}\sup_{\uu\in\RR^{\df}}\left\{\langle \uu, \uu^*\rangle - \langle \Phi, \mu \rangle -  \chi_{\MMM_\uu}(-\mu)\right\} \\ 
& = \sup_{\mu\in(\Mm_h)_+(\Omega)}\sup_{\uu\in\RR^{\df}}\{\langle \uu, \uu^*\rangle - \langle \Phi, \mu \rangle - \langle \uu + \Cb + \CC^*\mu , \uu^*\rangle 
+ \langle \uu + \Cb + \CC^*\mu , \uu^*\rangle  
\\ 
& \quad+ \langle (\Theta\bullet S)_T,\mu\rangle - \langle (\Theta\bullet S)_T,\mu\rangle
+\langle \lambda, \mu \rangle - \lambda - \langle \lambda, \mu \rangle + \lambda 
-  \chi_{\MMM_\uu}(-\mu)\} \\
& = \sup_{\mu\in(\Mm_h)_+(\Omega)}\{ L^\Theta_\lambda(-\uu^*,-\mu) + 
\sup_{\uu\in\RR^{\df}}\{ \langle \uu + \Cb - \CC^*(-\mu) ,\uu^*\rangle - \lambda 
+\langle\lambda +  (\Theta\bullet S)_T,-\mu\rangle -  \chi_{\MMM_\uu}(-\mu)\}\},  
\end{align*}
where~$L$ is the Lagrangian from~\eqref{eq:LagrangianFunction}, 
$\chi$ the indicator function, and we also used~\eqref{eq:Adefn}.
Hence the convex conjugate reads 
$\varpi^*(\uu^*) = \sup\left\{ L^\Theta_\lambda(-\uu^*,-\mu): -\mu\in\MMM_\uu \right\}$, and 
\begin{align*}
\varpi^{**}(\uu) &= \sup_{\uu^*\in\RR^{\df}} \left\{ \langle \uu, \uu^*\rangle - \varpi^*(\uu^*)\right\}
 =  \sup_{\uu^*\in\RR^{\df}} \inf_{-\mu\in\MMM_\uu} \left\{ \langle \uu, \uu^*\rangle - L^\Theta_\lambda(-\uu^*,-\mu) \right\} \\
& = - \inf_{\uu^*\in\RR^{\df}}\sup_{\mu\in\MMM_\uu} \left\{ \langle \uu, -\uu^*\rangle + L^\Theta_\lambda(-\uu^*,\mu) \right\}   
 = -\inf_{\uu^*\in\RR^{\df}}\sup_{\mu\in\MMM_\uu} \left\{ \langle \uu, \uu^*\rangle + L^\Theta_\lambda(\uu^*,\mu) \right\} =
-\val_{\Pp}(\uu).
\end{align*}
The Young-Fenchel inequality implies that $\varpi \geq \varpi^{**}$, 
and we recover weak duality between the primal~\eqref{eq:primalPerturbed} 
and the dual~\eqref{eq:dualPerturbed} problems: $\widetilde{\val}_p(\uu) \geq \widetilde{\val}_d(\uu)$.

By assumption there is no duality gap ($\widetilde{\val}_d(\uu) = \widetilde{\val}_p(\uu)$),
hence $\varpi(\uu) = \varpi^{**}(\uu)$ and~$\varpi$ 
is lower semi-continuous by Fenchel-Moreau Theorem~\cite[Section 31]{Rockafellar70}.
Since $\uu\in \interior \Mbf$, then~$\varpi$ is continuous at~$\uu$ 
by~\cite[Theorem~2.2.9]{Zalinescu02}.
By Proposition~\ref{prop:ValFnDifferential}(i) 
the sub-differential $\partial \varpi(\uu)$ is non-empty and by Proposition~\ref{prop:ValFnDifferential}(iii),
$\varpi$ is Hadamard directionally differentiable at~$\uu$ in any direction $\hh\in\RR^{\df}$, 
such that
$$
\varpi'(\uu,\hh) = \sup_{\uu^*\in\partial \varpi(\uu)}\left\langle \uu^*, \hh\right\rangle.
$$
Young-Fenchel inequality~\cite[Section 12]{Rockafellar70} then yields
$\varpi(\uu) = \left\langle \uu,\uu^*\right\rangle - \varpi^*(\uu^*)$
if and only if $\uu^*\in\partial \varpi(\uu)$ and hence 
it follows that $\varpi^{**}(\uu) = \varpi(\uu)$.
The primal problem~\eqref{eq:primalPerturbed} can be expressed as $-\varpi^{**}(\uu)$ 
by the discussion above and it is finite by assumption. 
Hence $\partial \varpi(\uu) = -\mathfrak{S}_{\uu}$ 
(the set of optimal solutions of the primal problem~\eqref{eq:primalPerturbed} at~$\uu$), and
$$
\varpi'(\uu,\hh) = \sup_{\uu^*\in -\mathfrak{S}_{\uu}}\left\langle \uu^*, \hh\right\rangle = 
- \inf_{\uu^*\in \mathfrak{S}_{\uu}}\left\langle \uu^*, \hh\right\rangle.
$$
The proposition then follows since $\varpi(\uu) = -\widetilde{\val}_d(\uu)$ and 
\begin{align*}
(\widetilde{\val}_d)'(\uu,\hh)
 &= \lim_{\eps\downarrow 0}\frac{\widetilde{\val}_d(\uu +\eps\hh) - \widetilde{\val}_d(\uu)}{\eps} 
 = \lim_{t\downarrow 0}\frac{-\varpi(\uu + \eps\hh) + \varpi(\uu)}{\eps} 
  = -\varpi'(\uu,\hh).
\end{align*}
\end{proof}
If the perturbation~$\uu$ is itself parametrised by a vector $\sss\in\RR^n$ for some $n<\infty$ and it is continuously differentiable with respect to this parameter then we have the following application of the Chain Rule~\ref{proposition:chainrule}. 
\begin{corollary}\label{corollary:superhedgingPerturbParam}
With the same assumptions as in Proposition~\ref{proposition:Sensitivity}, 
if $\uu := \uu(\sss)$ is continuously differentiable with respect some parameter $\sss\in\RR^n$, 
then the equalities
$$
(\widetilde{\val}_d \circ \uu)'(\sss,\hh)
 = \inf\left\{\left\langle \uu^*, \nabla\uu(\sss)\hh\right\rangle: \uu^*\in\widetilde{\mathfrak{S}}_{\uu}\right\}
\qquad\text{and}\qquad
(\utilde{\val}_d \circ \uu)'(\sss,\hh)
 = \sup\left\{\left\langle \uu^*, \nabla\uu(\sss)\hh\right\rangle: \uu^*\in\utilde{\mathfrak{S}}_{\uu}\right\}
$$
hold, where $\nabla\uu(\sss)$ is the Jacobian matrix evaluated at~$\sss$.
\end{corollary}
\begin{proof}
As $\uu$ is continuously differentiable it is Fr{\'e}chet differentiable and 
$(\uu)'(\sss,\hh) = \nabla \uu(\sss) \hh$. 
Since~$\widetilde{\val}_d$ is Hadamard differentiable at~$\uu$ by Proposition~\ref{proposition:Sensitivity}, 
the Chain Rule~\ref{proposition:chainrule} concludes the proof.
\end{proof}
If the super-hedging primal problem~\eqref{eq:primalPerturbed}
admits unique solutions at $\widetilde{\uu}_0\in\RR^{\df}$ and $\utilde{\uu}_0\in\RR^{\df}$,
then $\widetilde{\mathfrak{S}}_{\uu_0} = \{\widetilde{\uu}^*\}$ 
and~$\utilde{\mathfrak{S}}_{\uu_0} = \{\utilde{\uu}^*\}$ are singletons
and the derivatives in Proposition~\ref{proposition:Sensitivity} and Corollary~\ref{corollary:superhedgingPerturbParam} 
are linear in~$\hh$.
Thus, as in~\cite[Section 4.1]{GorbenaLopez14} 
there exist neighbourhoods
$\Bb_{\widetilde{\uu}_0}, \Bb_{\utilde{\uu}_0}\subset\RR^{\df}$ of~$\widetilde{\uu}_0$ 
and~$\utilde{\uu}_0$ 
such that for all $\uu\in \Bb_{\widetilde{\uu}_0}$ 
and all $\vv\in \Bb_{\utilde{\uu}_0}$ 
the values of the perturbed dual problems can be approximated as
$$
\widetilde{\val}_d(\uu)
 = \widetilde{\val}_d(\widetilde{\uu}_0) + \left\langle\widetilde{\uu}^*,\uu-\widetilde{\uu}_0\right\rangle + o(\uu-\widetilde{\uu}_0)
\quad\text{and}\quad
\utilde{\val}_d(\vv)
 = \utilde{\val}_d(\utilde{\uu}_0) + \left\langle\utilde{\uu}^*,\vv-\utilde{\uu}_0\right\rangle + o(\vv-\utilde{\uu}_0)
$$
This approximation can be naturally extended to the case where the perturbation is itself parametrised.
In particular for all~$\sss$ in the neighbourhood of~$\sss_0$, 
the approximation of the perturbed dual problem~\eqref{eq:dualPerturbed}
\begin{equation}\label{eq:valueEstimation}
\left\{
\begin{array}{rl}
\displaystyle \widetilde{\val}_d\circ\widetilde{\uu}(\sss)
 & = \displaystyle \widetilde{\val}_d \circ\widetilde{\uu}(\sss_0) + \left\langle\widetilde{\uu}^*, \nabla\widetilde{\uu}(\sss_0)(\sss-\sss_0)\right\rangle +o(\sss-\sss_0),\\
\displaystyle \utilde{\val}_d\circ\utilde{\uu}(\sss)
 & = \displaystyle \utilde{\val}_d \circ\utilde{\uu}(\sss_0) + \left\langle\utilde{\uu}^*, \nabla\utilde{\uu}(\sss_0)(\sss-\sss_0)\right\rangle +o(\sss-\sss_0).
\end{array}
\right.
\end{equation}

\section{Application to Forward-Start Straddle}\label{sec:applicationFwdStartStraddleCh4}
We perform a sensitivity analysis of the optimal values of robust hedging for Forward-Start Straddle 
with payoff $|S_{t_2} - \Kk S_{t_1}|$ 
for $0<t_1<t_2$ and various strikes $\Kk>0$,
with respect to extrapolation of the total implied variance at~$t_1$ and~$t_2$. 
We assume that the primal perturbed problem~\eqref{eq:primalPerturbed} 
admits a unique solution,
and consider as inputs 
Calls maturing at~$t_{1}$ with strikes $K^{t_{1}}_1,\ldots,K^{t_{1}}_{\kappa(t_{1})}$,
and Calls maturing at~$t_{2}$ with strikes $K^{t_{2}}_1,\ldots,K^{t_{2}}_{\kappa(t_{2})}$,
with $\kappa(t_{1}), \kappa(t_{2})$ both finite. 
The vector of normalised Calls then reads 
\begin{equation}
\Cb = \left(c(K^{t_{1}}_1,t_{1}),\ldots,c(K^{t_{1}}_{\kappa(t_{1})},t_{1}),c(K^{t_{2}}_1,t_{2}),\ldots,c(K^{t_{2}}_{\kappa(t_{2})},t_{2}) \right).
\end{equation}
We parametrise the total implied variance surface $w$ by a vector of parameters $\sss\in\RR^l$
such that that the resulting surface is arbitrage free and grows at most linearly in the wings,
and we denote it by~$w(\cdot,\cdot;\sss)$. 
\begin{assumption}\label{assumption:VarianceParamDifferentiation}
The parametrisation $w(\cdot,\cdot;\sss)$ is continuously differentiable with respect to~$\sss$.
\end{assumption}
We can then calculate the resulting total implied volatility 
$I_{i}^{t}(\sss) := \sqrt{w(k_{i}^{t},t;\sss)}$, 
where $k = \log(K)$,
and define the vector of perturbed prices as
\begin{align*}
\Cb(\sss) := \Cb + \uu(\sss) := \left(\crbs_1^{t_{1}}(\sss),\ldots,\crbs_{\kappa(t_{1})}^{t_{1}}(\sss),
\crbs_{1}^{t_{2}}(\sss),\ldots,\crbs_{\kappa(t_{2})}^{t_{2}}(\sss)\right),
\end{align*}
where for simplicity
$\crbs_{i}^{t}(\sss) := c_{\BS}(k^{t}_i,I_{i}^{t}(\sss))$
for $t\in\{t_{1}, t_{2}\}$, $i=1,\ldots, \kappa(t)$.
We can compute sensitivities of perturbed prices with respect to~$\sss$. 
\begin{lemma}\label{lemma:sensitivityToParamsCalls}
For any $t\in\{t_{1}, t_{2}\}$, $i=1,\ldots, \kappa(t)$, $j=1,\ldots,l$,
$\Vr_{i}^{t}(\cdot)$ denoting the Black-Scholes Vega,
\begin{equation}\label{eq:CallPricePerturbSens}
\frac{\partial \crbs_{i}^{t}(\sss)}{\partial \sss_j}
 = \frac{\Vr_{i}^{t}(\sss)}{2I_{i}^{t}(\sss)\sqrt{t}}\frac{\partial w(k^{t}_i,t;\sss)}{\partial \sss_j}.
\end{equation}
\end{lemma}
\begin{proof}
A simple application of the chain rule together with Assumption~\ref{assumption:VarianceParamDifferentiation} yields, for $t\in\{t_1, t_2\}$, 
$$
\frac{\partial \crbs_{i}^{t}(\sss)}{\partial \sss_j} 
= \Vr_{i}^{t}(\sss)\frac{\partial I_{i}^{t}(\sss)}{\partial \sss_j}
 = \Vr_{i}^{t}(\sss)\frac{\partial w(k^{t}_i,t;\sss)}{\partial \sss_j}\frac{\rd I_{i}^{t}(\sss)}{\rd w(k^{t}_i,t;\sss)}
= \frac{\Vr_{i}^{t}(\sss)}{2I_{i}^{t}(\sss)\sqrt{t}}\frac{\partial w(k^{t}_i,t;\sss)}{\partial \sss_j}.
$$
\end{proof}
The Jacobian matrix of the perturbed Call prices then reads 
$$
\nabla \Cb(\sss) := \begin{pmatrix} 
 \partial_{\sss_1}\crbs_{1}^{t_1}(\sss) & \ldots & \partial_{\sss_l}\crbs_{1}^{t_1}(\sss)\\ 
 \vdots & \ddots  & \vdots \\ 
\partial_{\sss_1}\crbs_{\kappa(t_{1})}^{t_{1}}(\sss) & \ldots & \partial_{\sss_l}\crbs_{\kappa(t_{1})}^{t_{1}}(\sss) \\ 
\partial_{\sss_1}\crbs_{1}^{t_{2}}(\sss) & \ldots & \partial_{\sss_l}\crbs_{1}^{t_{2}}(\sss)\\ 
\vdots & \vdots & \vdots \\ 
\partial_{\sss_1}\crbs_{\kappa(t_{2})}^{t_{2}}(\sss) & \ldots  & \partial_{\sss_l}\crbs_{\kappa(t_{2})}^{t_{2}}(\sss)
\end{pmatrix}
\in  \mathscr{M}_{\kappa(t_1)+\kappa(t_2), l}(\RR),
$$
where $\mathscr{M}_{\kappa(t_1)+\kappa(t_2), l}(\RR)$ is the space of matrices of size $(\kappa(t_1)+\kappa(t_2))\times l$ with real entries.
For the numerics, we consider $t_1=1$ year and $t_2=1.5$ years;
the set of trading strategies is discretised using a monomial basis of degree at most~$4$ 
and there are~$18$ options available for each maturity for static hedging with moneyness in 
$\{0.3,0.4,0.5,\ldots,2.0\}$. 
However we assume that only a subset of those options 
has quotable market prices and the rest are priced by extrapolating the total implied variance. 
The state space is taken to be $[0,5]\times[0,5]$ with $500$ discretisation points for both maturities.

\subsection{Application to the Black-Scholes model}\label{sec:BScase}
If only prices of at-the-money Call options are observable for each maturity,
it is not unreasonable to fit the Black-Scholes model
$\rd S_t = \Sigma S_t\rd W_t$ ($S_0=1$). 
The only parameter that needs calibration is~$\Sigma$,
and we let $\Sigma = 20\%$.
The resulting total implied variance function $w : \RR\times \Tt \to \RR_+$ is constant in log-moneyness 
for each maturity and $w(\cdot,t) = \Sigma^2 t$ for $t\in\Tt$. 
Assume now that the actual shape of the total implied variance for each $t\in\Tt$ is
\begin{equation}\label{eq:BSextrapolation}
w(k,t;\sss) = p_t|k |+ \Sigma^2t,
\end{equation}
where $p_t\in\RR$ is the symmetric slope on both sides of the smile, 
so that $\sss = (p_{t_{1}},p_{t_{2}})\in\RR^2$.
For each $t\in\Tt$, the function~$\gr$ in~\eqref{eq:functiong} must be non-negative on $(k^*_t, \infty)$, 
which, by Lemma~\ref{lem:LinearVarNoArb}, is equivalent to $p_t\in [0,2]$ 
and the existence of a $k^*_t\in\RR_+$ as in the lemma. 
As we propose extrapolation of the total implied variance to the right on $(0,\infty)$
and to the left on $(-\infty, 0)$, then $k^*_t = 0$ (as $\gr(0) = \Sigma^2t > 0$), 
which places further restrictions on~$p_t$. 
In particular if $\Sigma^2t \geq 2 - \sqrt{2 - p_t^2}$ then $\gr(k)\geq 0$ for all $k>0$ 
by Lemma~\ref{lem:LinearVarNoArb}. 
This inequality places an upper bound on~$p_t$ for each $t\in\Tt$ such that any extrapolation 
with slope satisfying this bound is free of arbitrage. 
If $\Sigma^2t < 2 - \sqrt{2 - p_t^2}$ then 
\begin{equation}
\gr(k) > 0, 
\quad 
\text{ for all }
\quad 
k>\frac{p_t^2(\Sigma^2t + 2) - 8\Sigma^2t + 2p_t\sqrt{\Sigma^4t^2 - 4\Sigma^2t + p_t^2}}{p_t(4 - p_t^2)}.
\end{equation}
It follows that the proposed extrapolation~\eqref{eq:BSextrapolation} is arbitrage free 
if the expression on the right-hand side is equal to zero. 
The resulting quartic equation in~$p_t$ does not have real roots for either $t\in\Tt$ 
when $\Sigma = 0.2$ and $\Tt = \{1, 1.5\}$.
Hence the only viable values for~$p_t$ are between $0$ and $\sqrt{4 - (2-\Sigma^2t)^2}$ for each $t\in\Tt$ 
(where the upper bound is obtained by solving the quadratic equation $\Sigma^2t = 2 - \sqrt{2 - p_t^2}$). 
\begin{assumption}
Both slopes are equal: $p_{t_{1}} = p_{t_{2}} = a$.
\end{assumption}
This assumption could be relaxed, but at the cost of checking absence of calendar spread arbitrage 
$\partial_t w(k,t)\geq 0$~\cite[Lemma 2.1]{GatheralJacquier14}. 
Therefore a potential choice for the slopes would be to increase the value of the slope for each wing as maturity increases. 
The Jacobian now reads 
$$
\nabla \Cb(\sss) = \begin{pmatrix} 
\partial_{p_{t_1}} \crbs_{1}^{t_1}(\sss)) & 0 \\ 
\vdots  & \vdots \\ 
\partial_{p_{t_1}}\crbs_{\kappa(t_1)}^{t_1}(\sss) & 0 \\ 
0 & \partial_{p_{t_2}}\crbs_{1}^{t_2}(\sss)) \\ 
\vdots & \vdots \\ 
0  & \partial_{p_{t_2}}\crbs_{\kappa(t_2)}^{t_2}(\sss)
\end{pmatrix},
$$
and by Lemma~\ref{lemma:sensitivityToParamsCalls} and~\eqref{eq:BSextrapolation},
 we obtain, for each $t\in\Tt$, $i=1,\ldots,\kappa(t)$,
$\displaystyle \frac{\partial \crbs_{i}^{t}(\sss)) }{\partial p_t} = 
\frac{\Vr_{i}^{t}(\sss)|k^{t}_i|}{2I_{i}^{t}(\sss)\sqrt{t}}$.
Below we present numerical results for the super- and sub-hedging primal programmes for the at-the-money Forward-Start Straddle $\Kk = 1$. 
Tables~\ref{table:BSSuperHedgePerturbation} and~\ref{table:BSSubHedgePerturbation} 
summarise the results of the perturbation analysis for the super- and sub-hedging problems introduced above.
The column `Perturbation' contains the values of the slopes of extrapolation of the total implied variance. 
As expected the optimal values of the perturbed problems converge to the optimal value 
of the unperturbed problem in the first row. 
The column `Est. Value' contains the first-order expansion~\eqref{eq:valueEstimation},
and the last column is the absolute difference between the optimal value of the perturbed problem 
obtained by solving~\eqref{eq:primalPerturbed} and the value of the programme estimated 
via~\eqref{eq:valueEstimation}. 
The estimation becomes increasingly better the smaller the perturbation becomes. 
It confirms that the perturbation results presented in Section~\ref{sec:perturbationAnalysis} are local in nature.  
\begin{table}[h!tb]
\begin{center}
	\begin{tabular}{| c | c | c | c |}
	\hline
	Perturbation &	Optimal Value & Est. Value & Abs. Diff. \\ \hline
	0 & 0.149	& 0.149	& 0 \\ \hline
	5E-05 &	0.149 & 0.149 &	2.98E-10 \\ \hline
	1E-04 & 0.1490	& 0.149 & 1.19E-08 \\ \hline
	5E-03 &	0.1496 &  0.1496 &	1.57E-06 \\ \hline
	0.0476 & 0.1544 & 0.1552 &	7.75E-04 \\ \hline
	0.202	&	0.1563 & 0.1753 & 1.9E-02 \\ \hline
	\end{tabular}
\end{center}	
\caption{Perturbation of the super-hedging primal problem for the ATM Forward-Start Straddle in the Black-Scholes case.}
\label{table:BSSuperHedgePerturbation}
\begin{center}
	\begin{tabular}{| c | c | c | c |}
	\hline
	Perturbation &	Optimal Value & Est. Value & Abs. Diff. \\ \hline
	0 & 0.0385	& 0.0385	& 0 \\ \hline
	5E-05 & 0.0385 & 0.0385 &	2.88E-07 \\ \hline
	1E-04 & 0.0385 & 0.0385 & 3.42E-07 \\ \hline
	5E-03 &	0.0383 &  0.0383 & 1.16E-05 \\ \hline
	0.0476 & 0.0365 & 0.0359 & 6.1E-04 \\ \hline
	0.202	&	0.0357 & 0.0272 & 8.53E-03 \\ \hline
	\end{tabular}
\end{center}	
\caption{Perturbation of the sub-hedging primal problem for the ATM Forward-Start Straddle in the Black-Scholes case.}
\label{table:BSSubHedgePerturbation}	
\end{table}

\subsection{Application to the Heston model}\label{sec:HestonCaseChapter4}
Assume now that for each maturity, only Call options with moneyness in $\Kf :=\{0.8,0.9,\ldots,1.2\}$ are traded,
and that observed prices are consistent with the Heston stochastic volatility model~\cite{Heston}, 
where the stock price process is the unique strong solution to
\begin{equation}\label{eq:HestonSDE}
\begin{array}{rcll}
\rd S_t & = & S_t \sqrt{V_t} \rd W_t, & S_0 = 1,\\
\rd V_t & = & \kappa\left(\theta-V_t\right)\rd t + \xi\sqrt{V_t} \rd Z_t, & V_0 = v>0,
\end{array}
\end{equation}
where $W$ and $Z$ are two one-dimensional standard Brownian motions with
$\rd\langle W, Z\rangle_t = \rho \rd t$, $\kappa, \theta, \xi>0$ and $\rho \in [-1,1]$.
We consider here $(\kappa, \theta, \xi, v, \rho) = (1, 0.07, 0.4, 0.07, -0.8)$.
In principle calibrating Heston provides an extrapolation of the total implied variance, 
however there is no closed-form expression, and thus we make a simplifying assumption 
on the extrapolation of the implied variance beyond observable strikes. 
We assume that the total implied variance is extrapolated linearly to the left 
and to the right of the last observed strike for each maturity $t\in\Tt$. 
Let $L:=\min\{i=1,\ldots,18: K_L = \min \Kf\}$ 
and $R:= \max\{i=1,\ldots,18: K_R = \max \Kf_{\mathrm{market}}\}$
denote the smallest and largest indices at which the options are quoted.
Then for a vector $\sss := (q_{t_1},p_{t_1},q_{t_2},p_{t_2})$, 
the wing extrapolations read, for $t\in \Tt$,
\begin{equation}\label{eq:WingHeston}
w(k,t;\sss) = \left\{
\begin{array}{ll}
\psi(q_{t})|k-k_{L}| + w(k_{L},t), & \text{for } k\leq k_{L},\\
\psi(p_{t})|k-k_{R}| + w(k_{R},t), & \text{for } k\geq k_{R},
\end{array}
\right.
\end{equation}
where $\psi(z):= 2-4(\sqrt{z(z+1)}-z)$ as introduced \cite{Lee} and discussed above.
The Jacobian reads 
$$
\nabla \Cb(\sss)
 = \begin{pmatrix} 
 	\crbs'_{1}(\sss) & \mathrm{O}_{L-1} & \mathrm{O}_{L-1} &\mathrm{O}_{L-1}\\
	- & - & - & -\\ 
	\mathrm{O}_{\kappa(t_1)-R} & \crbs'_{2}(\sss) & \mathrm{O}_{\kappa(t_1)-R}  & \mathrm{O}_{\kappa(t_1)-R}\\
	\mathrm{O}_{L-1} & \mathrm{O}_{L-1} & \crbs'_{3}(\sss) & \mathrm{O}_{L-1} \\
	- & - & - & -\\ 
	\mathrm{O}_{\kappa(t_2)-R} & \mathrm{O}_{\kappa(t_2)-R} & \mathrm{O}_{\kappa(t_2)-R} & \crbs'_{4}(\sss)
\end{pmatrix},
$$
where the dashed lines are null matrices of size $(R-L+1, 4)$ and correspond to the initial (unperturbed) inputs,
the~$\mathrm{O}$ are null column vectors with size in subscript, and the~$\crbs'(\sss)$ are column vectors of derivatives:
\begin{equation*}
\begin{array}{rlcl}
\crbs'_{1}(\sss) & := \displaystyle (\partial_{q_{t_{1}}}\crbs_{i}^{t_{1}}(\sss))_{i=1,\ldots,L-1}
 & \text{ and } & 
\crbs'_{3}(\sss) := \displaystyle (\partial_{q_{t_{2}}}\crbs_{i}^{t_{2}}(\sss))_{i=1,\ldots,L-1},\\
\crbs'_{2}(\sss) & := \displaystyle (\partial_{p_{t_1}} \crbs_{i}^{t_1}(\sss))_{i=R+1,\ldots,\kappa(t_1)}
 & \text{ and } & 
\crbs'_{4}(\sss) := \displaystyle (\partial_{p_{t_2}} \crbs_{i}^{t_2}(\sss))_{i=R+1,\ldots,\kappa(t_2)}.
\end{array}
\end{equation*}
Note that rows of zeros correspond to sensitivities of the traded Call option prices, which naturally do not depend on the extrapolation of the wings. 
\begin{lemma}\label{lemma:varExtrapHestonSens}
For $w(k,t;\sss)$ in~\eqref{eq:WingHeston} 
for each $t\in\Tt$, $k\in\RR$, 
the following holds for $i=1,\ldots,4$:
$$
\frac{\partial w(k,t;\sss)}{\partial \sss_i} = - \frac{\left|k - \ind_{\{i=1,3\}}(i)k_L - \ind_{\{i=2,4\}}(i)k_R\right|\psi(\sss_i)}{\sqrt{\sss_i(1+\sss_i)}}.
$$
\end{lemma}
\begin{proof}
The chain rule yields
$$
\frac{\partial w(k,t;\sss)}{\partial \sss_i} = \left|k - \ind_{\{i=1,3\}}(i)k_L - \ind_{\{i=2,4\}}(i)k_R\right|\frac{\partial \psi(\sss_i)}{\partial \sss_i},
$$
and 
$$
\frac{\partial \psi(\sss_i)}{\partial \sss_i} = \frac{\partial}{\partial \sss_i}
\left[2 - 4\left(\sqrt{\sss_i(1+\sss_i)} - \sss_i\right)\right]
= 
\frac{4\left(\sqrt{\sss_i(1+\sss_i)} - \sss_i\right) - 2}{\sqrt{\sss_i(1+\sss_i)}} = -\frac{\psi(\sss_i)}{\sqrt{\sss_i(1+\sss_i)}}.
$$
\end{proof}
Then by Lemmas~\ref{lemma:sensitivityToParamsCalls} and~\ref{lemma:varExtrapHestonSens} we have, for all $j=1,\ldots 4$,  $i=1,\ldots,\kappa(t)$ and $t\in\Tt$,
$$
\frac{\partial \crbs_{i}^{t}(\sss)) }{\partial \sss_j} = 
-\frac{\Vr_{i}^{t}(\sss)\left|k - \ind_{\{j=1,3\}}(i)k_L - \ind_{\{j=2,4\}}(i)k_R\right|\psi(\sss_j)}{2I_{i}^{t}(\sss)\sqrt{\sss_j(1+\sss_j)t}}.
$$
As discussed in~\cite[Section 6.3]{BenaimFriz08},
the slope of the total implied variance for a fixed~$t$ 
as~$k$ tends to infinity is equal to~$\psi(p^{*})$ where~$p^{*}$ is a root of a non-linear equation 
\begin{equation}\label{eq:exponentPEq}
(\kappa - \rho\xi p^{*})^2
+ \left(\xi^2 p^{*}(p^{*}-1) - (\kappa - \rho\xi p^{*})^2\right)^{1/2}\cot\left(\frac{\sqrt{\xi^2p^{*}(p^{*}-1) - (\kappa - \rho\xi p^{*})^2}t}{2}\right) = 0. 
\end{equation}
We can use the above equation to calculate the slope of the left wing of a slice of the total implied variance as~$k\downarrow-\infty$. 
The symmetric process~$1/S$ follows the same SDE~\eqref{eq:HestonSDE} with amended parameters:
with $X := \log(S)$ and $Y = -X$, 
It\^{o}'s lemma implies 
$\rd X_t = -\frac{1}{2}V_t\rd t + \sqrt{V_t}\rd W_t$ 
and 
$\rd Y_t = \frac{1}{2}V_t\rd t + \sqrt{V_t}\rd B_t$,
where $\rd B_t := \sqrt{V_t}\rd t - \rd W_t$ is a Brownian motion with drift. 
Also note that $Z = \rho W + \sqrt{1 - \rho^2}W^1$,
where~$W$ and~$W^1$ are independent. 
Therefore
$$
\rd Z_t = \rho \left(\sqrt{V_t}\rd t - \rd B_t\right) + \sqrt{1 - \rho^2}W^1_t = 
\rho\sqrt{V_t}\rd t + \rd W^2_t, 
$$
where $W^2_t := -\rho B_t + \sqrt{1 - \rho^2}W^1_t$ and the instantaneous variance~$V$ satisfies
$\rd V_t = \widetilde{\kappa}(\widetilde{\theta} - V_t)\rd t + \xi \sqrt{V_t}\rd W^2_t$,
with $\widetilde{\kappa} := \kappa - \rho\xi$ and $\widetilde{\theta} := \kappa\theta/(\kappa - \rho\xi)$.
Thus the inverse of~$S$ follows~\eqref{eq:HestonSDE} 
with parameters $\widetilde{\kappa}, \widetilde{\theta}, \xi > 0$ and $\widetilde{\rho} := -\rho \in [-1, 1]$ 
only if $\kappa > \rho\xi$, which is automatically satisfied as $\rho < 0$ in our case. 
As the higher moments of~$1/S$ are the negative moments of~$S$, 
the parameter~$q^*$ of the slope~$\psi(q^*)$ of the left wing can be calculated as a solution of the non-linear equation~\eqref{eq:exponentPEq} with parameters $\widetilde{\kappa}$ and $\widetilde{\rho}$ 
substituted instead of~$\kappa$ and~$\rho$.
Thus we can calculate the vector~$\sss$ using~\eqref{eq:exponentPEq} and the discussion above. 

Table~\ref{table:PerturbationT1} presents the sets of slopes used to extrapolate the total implied variance for both maturities. 
The perturbation sets are numbered for ease of reference and Set~$1$ corresponds to the unperturbed case.
The parameters in this set are calculated by solving~\eqref{eq:exponentPEq}.
As discussed in the Black-Scholes case in Section~\ref{sec:BScase},
other perturbation sets were chosen so that the slices of the total implied variance do not cross.
Tables~\ref{table:HestonSuperHedgePerturbation} and~\ref{table:HestonSubHedgePerturbation} 
show the perturbation analysis for the super- and sub-hedging problems respectively.
As in the Black-Scholes case in Section~\ref{sec:BScase} the results are in line with expectations, as the approximation becomes less accurate as the perturbation parameters deviate from the unperturbed case 
(presented in the first row of each Table). 
It also confirms that the perturbation results obtained in Section~\ref{sec:perturbationAnalysis} are local in nature.
It must be noted that the results in Black-Scholes and Heston 
imply that the at-the-money Forward-Start Straddle is not very sensitive to errors in extrapolation of the spot total implied variance. 
In particular, even if the extrapolation is very inaccurate, the price of Forward-Start options close to at-the-money will not vary significantly.
These confirm the results obtained in~\cite{BJLR17} in the sense that
European options cannot effectively hedge forward volatility claims, 
and instead Forward-Start options should be viewed as 
input (when traded liquidly) into the calibration of forward volatility-dependent exotics.

\begin{table}[h!t]
\begin{center}
	\begin{tabular}{| c | c | c | c | c | c | c |}
	\hline
	{Perturbation Set} &	1 & 2 & 3 & 4 & 5 & 6 \\ \hline
	$q_{t_1}$ & 5.058 & 5.06 & 5.2 & 6 & 10 & 12 \\ \hline
	$p_{t_1}$ &  24.21 & 24.22  & 24.35 & 25.1 & 35 & 37 \\ \hline
	$\psi(q_{t_1})$ & 0.0901 & 0.09011 & 0.0879 & 0.077 & 0.0476 & 0.04 \\ \hline
	$\psi(p_{t_1})$ &  0.0202 & 0.02022 & 0.0201 & 0.0195 & 0.0141 & 0.0133 \\ \hline
	$q_{t_2}$ & 6.83 &  6.84	& 6.9 &	7.1	& 10 &	12 \\ \hline
	$p_{t_2}$ &  30.714 & 30.72 & 30.73 & 31.1 & 35 & 37 \\ \hline
	$\psi(q_{t_2})$ & 0.0683 &  0.0682 & 0.0676 & 0.0659 & 0.0476 &	0.04 \\ \hline
	$\psi(p_{t_2})$ &  0.016 & 0.01601 & 0.016 & 0.0158 & 0.0141 & 0.0133 \\ \hline
	\end{tabular}
\end{center}	
\caption{Perturbation parameters and corresponding total implied variance slopes.}
\label{table:PerturbationT1}
\end{table}
\begin{table}[h!t]
\begin{center}
	\begin{tabular}{| c | c | c | c |}
	\hline
	Perturbation set & Optimal Value & Est. Value & Abs. Diff. \\ \hline
	1 & 0.1616 & 0.1616	        & 0 \\ \hline
	2 & 0.1616 & 0.1616 &	2.77E-08 \\ \hline
	3 & 0.1617	& 0.1617  & 5.16E-06 \\ \hline
	4 & 0.1624 &  0.1627  &	2.38E-04 \\ \hline
	5 & 0.1627 & 0.1654  &	2.65E-03 \\ \hline
	6 & 0.1625 & 0.1662  & 3.69E-03 \\ \hline
	\end{tabular}
\end{center}	
\caption{Perturbation of the super-hedging primal problem for the ATM Forward-Start Straddle in Heston.}
\label{table:HestonSuperHedgePerturbation}	
\begin{center}
	\begin{tabular}{| c | c | c | c |}
	\hline
	Perturbation set & Optimal Value & Est. Value & Abs. Diff. \\ \hline
	1 & 0.04455	& 0.04455	& 0 \\ \hline
	2 & 0.04455 & 0.04455   & 6.78E-09\\ \hline
	3 & 0.04452 & 0.04451 & 4.83E-06 \\ \hline
	4 &  0.04437  &  0.04427 & 1.06E-04 \\ \hline
	5 & 0.04432 & 0.04353 & 7.90E-04 \\ \hline
	6 &	0.04436 & 0.04329 & 1.07E-03 \\ \hline
	\end{tabular}
\end{center}	
\caption{Perturbation of the sub-hedging primal problem for the ATM Forward-Start Straddle in Heston.}
\label{table:HestonSubHedgePerturbation}	
\end{table}

\newpage
\appendix
\section{Cones and directional derivatives}\label{sec:AppReminder}
Let~$\Xx$ be a normed topological vector space, 
and $\Xx^*$ its topological dual space. 
We first recall several facts about Riesz spaces and convex cones in vector spaces,
taking~\cite{AB07} as our main guide.
\begin{definition}\cite[Section 8.1]{AB07}\label{defn:cones}
A positive convex cone $X_{+}\subset\Xx$ is closed under operations of addition and multiplication by a non-negative real-valued scalar together with the property $X_{+}\cap (-X_{+}) = \{0\}$. 
A strictly positive cone~$X_{++}$ is defined as $X_{++} := X_+\setminus\{0\}$.
\end{definition}
For every application in this paper, $\Xx$ is endowed with a partial order induced by a positive convex cone $X_+\in\Xx$,
i.e. for any two elements $x_1, x_2\in\Xx$ we have $x_1\geq x_2$ if and only if $x_1 - x_2 \in X_+$.
If for any two elements $x_1, x_2\in \Xx$ their minimum $x_1\wedge x_2$ and maximum $x_1\vee x_2$ also belong to $\Xx$ then it is a Riesz space~\cite[Section 8.2]{AB07}. 
In a Riesz space~$\Xx$, order unit elements play a special role: 
\begin{definition}\cite[Section 8.7]{AB07}\label{defn:OrderUnit}
An element $u\in X_{++}$ is called an order unit if for all $x\in X$ 
there exists $\lambda > 0$ such that $-\lambda u \leq x\leq \lambda u$.
\end{definition}
If the Riesz space $\Xx$ is norm-complete then it becomes a Banach lattice, an important subset of locally convex topological Riesz spaces. 
\begin{definition}\cite[Section IV.3, Definition 3.2]{Bichteler98}\label{defn:BanachLattice}
If a Riesz space $\Xx$ is endowed with a norm $\|\cdot\|_{\Xx}$ that makes it complete then it is called a Banach lattice.
\end{definition}
As $\Xx$ admits a topological dual $\Xx^*$ we can define dual sets to the positive convex cone $X_+$. 
\begin{definition}\label{defn:NegativePolar}
The negative polar 
$X^*_+ := \left\{ x^* \in \Xx^* : \left\langle x,x^*\right\rangle \geq 0\text{ for all }x\in X_+\right\}$
is the dual of~$X_+$.
\end{definition}

We now recall some useful notions on directional derivatives for convex functions
needed for the perturbation analysis
in Section~\ref{sec:PerturbationAnalysis}.
Let $g: \Xx\to \overline{\RR}$ 
an extended real-valued function.
\begin{definition}\cite[Definition 2.45]{BonnansShapiro00}
The mapping~$g$ is directionally differentiable at $x\in \Xx$ in the Hadamard sense if 
the directional derivative~$g'(x,\hh)$ exists for all $\hh\in\Xx$ and the equality
$$
g'(x,\hh) = \lim_{n\uparrow\infty}\frac{g(x + \eps_n \hh_n) - g(x)}{\eps_n}
$$
holds for any sequences $(\hh_n)_{n\in\mathbb{N}}\in\Xx$ converging to~$\hh$ and 
$(\eps_n)_{n\in\mathbb{N}} \in\RR$ converging to zero.
In addition if $g'(x,\cdot)$ is linear in~$\hh$ then it is said to be Hadamard differentiable at~$x$.
\end{definition}
If~$g$ is Hadamard differentiable at $x\in \Xx$ then~$g'(x,\cdot)$ 
is continuous on~$\Xx$~\cite[Proposition~2.46]{BonnansShapiro00}.
Hadamard differentiability though, is a more restrictive notion of directional differentiability, 
as opposed, for example, to Fr{\'e}chet differentiability~\cite[Definition~A.1]{AndrewsHopper11}. 
Nonetheless the following holds:
\begin{proposition}\cite[Proposition~2.49]{BonnansShapiro00}
If~$g$ is directionally differentiable at $x$ and Lipschitz continuous (with constant~$L$) 
in a neighbourhood of~$x$,
then it is directionally differentiable at~$x$ in the Hadamard sense and the directional derivative $g'(x,\cdot)$ 
is Lipschitz continuous (with same constant~$L$) on~$\Xx$.
\end{proposition}
If $\Xx$ is a finite dimensional, then the situation simplifies considerably. 
If $g$ is also locally Lipschitz continuous at $x\in \Xx$ then the Hadamard and the Fr{\'e}chet derivatives are equivalent.
In particular all proper convex functions are locally Lipschitz (\cite[Proposition~2.107]{BonnansShapiro00}) and if the underlying space is finite-dimensional then they are continuous 
on the relative interior of their effective domains~\cite[Theorem~7.24]{AB07}.
We now state some technical results needed in the paper.
\begin{proposition}(Chain rule~\cite[Proposition~2.47]{BonnansShapiro00})\label{proposition:chainrule}
If $g : \Xx \to \Yy$ is directionally differentiable at~$x$ and $f : \Yy \to Z$ is Hadamard differentiable at $y=g(x)$, then $f \circ g$ is directionally differentiable at~$x$ and
$(f\circ g)'(x,\hh) = f'(y,g'(x,\hh))$.
Moreover if~$g$ (resp. $f$) is Fr{\'e}chet differentiable at~$x$ (resp.~$y$), then $f \circ g$ is Fr{\'e}chet differentiable at~$x$. 
\end{proposition}
\begin{proposition}\cite[Proposition~2.126 (iv-v)]{BonnansShapiro00}\label{prop:ValFnDifferential}
If $\Xx$ is a Banach space endowed with the norm topology
and $g: \Xx\to\RR$ is convex and continuous at $x\in \Xx$, then
\begin{enumerate}[(i)]
\item $g$ is sub-differentiable at $x$;
\item $\partial g(x)$ is a non-empty, convex and weak* compact subset of $\Xx^*$; 
\item $g$ is Hadamard directionally differentiable at~$x$ and, for any $\hh\in \Xx$,
$g'(x,\hh) = \sup_{x^*\in\partial g(x)}\left\langle x^*,\hh\right\rangle$.
\end{enumerate}
\end{proposition}
Of course, if $\partial g(x) = \{a\}$, 
then $g'(x,\hh) = \langle a,\hh\rangle$ and $g$ is Hadamard differentiable at~$x$.
Similar results are proved in~\cite[Theorem~23.4]{Rockafellar70} when $\Xx$ is a finite-dimensional vector space.

\section{Proofs}

\subsection{Proof of Theorem~\ref{thm:FTAPnoOrderUnit}}\label{sec:thmFTAPnoOrderUnit}
Suppose there exists a strictly positive linear functional $\pi : \Cc_{h}(\Omega) \to \RR$ that extends~$\rho$. 
As~$\Cc_{h}(\Omega)$ is a Banach lattice, then~$\pi$ is continuous by~\cite[Theorem 1.36]{AliTourky}.
It is also evident that it implies absence of weak free lunch. 
Conversely, assume that there is no weak free lunch. 
It then follows that $m_0\notin\overline{\Fr-(\Cc_{h})_+(\Omega)}$ by Assumption~\ref{assumption:risklessBond}.
As~$\{m_0\}$ is compact and $\overline{\Fr-(\Cc_{h})_+(\Omega)}$ is closed in the weak topology, 
the Strong Separating Hyperplane Theorem~\cite[Theorem~5.79]{AB07}
implies that there exists a non-zero continuous linear functional $\pi : \Cc_{h}(\Omega) \to \RR$ such that $\pi(m_0) > 0$
and $\pi(f-g)\leq 0$ for all $f\in \Fr$ and $g\in (\Cc_{h})_+(\Omega)$.  
As $0\in \Fr$ it follows that $\pi(-g)\leq 0$ for all $g\in (\Cc_{h})_+(\Omega)$ and hence $\pi$ is positive. 
Moreover $\pi(g) > 0$ for all $g\in (\Cc_{h})_{++}(\Omega)$. 
Otherwise there exists $g\in (\Cc_{h})_{++}(\Omega)$ such that $\pi(g) = 0$, 
i.e. $g\in\Fr$ and hence $g\in \overline{\Fr-(\Cc_{h})_+(\Omega)}\cap (C_h)_+(\Omega)$
which contradicts the absence of weak free lunch.
Similarly as $0\in (\Cc_{h})_+(\Omega)$ one has $\pi(f)\leq 0$ for all $f\in \Fr$. 
Therefore there exists $\xi\in\RR$ such that $\xi\pi(m) = \rho(m)$ for all $m\in \Mf$.
As $\xi\pi(m_0) = \rho(m_0) > 0$ implies that $\xi > 0$ and without loss of generality one can take $\xi = 1$. 
Thus we have shown existence of a strictly positive continuous and linear functional 
$\pi : \Cc_{h}(\Omega) \to \RR$ that extends~$\rho$.

Let us define a map $T:\Cc_{h}(\Omega) \to \Cc_{b}(\Omega)$ such that $T(f) := f/h$ and note that it is an isometry.
Define a functional $\widetilde{\pi} : \Cc_b(\Omega) \to \RR$ by $\widetilde{\pi}(f) :=  C\pi(T^{-1}(f))$ for all $f\in \Cc_b(\Omega)$, 
where $C$ is a positive real constant. 
Note that $\widetilde{\pi}$ is continuous, linear and strictly positive by definition. 
The space $\Cc_b(\Omega)$ can be identified with $\breve{\Cc}(\Omega)$, the space of continuous functions on $\breve{\Omega}$ which is the the Stone-\v{C}ech compactification of~$\Omega$.
As the dual of $\Cc(\breve{\Omega})$ can be identified with the space of regular signed Borel measures of bounded variation~\cite[Theorem~14.12]{AB07}, the following representation holds:
$$
\widetilde{\pi}\circ T(f) = \int_{\breve{\Omega}} \breve{T}(f)(\omega)\nu(\rd \omega),
$$
where $\breve{T}$ is the unique extension of $T(f) \in \Cc_b(\Omega)$.
Note that since~$\Cc_b(\Omega)$ is locally compact, we could avoid Stone-\v{C}ech compactification arguments, 
using~\cite[Theorem~7.11.3]{Bogachev2}. Since our setup was inspired by~\cite{ABPS13}, we instead followed their steps to prove our statement.
Observe that $\nu$ is positive as $0 < C\pi(f) = \widetilde{\pi} \circ T(f) = \int_{\breve{\Omega}} \breve{T}(f)(\omega)\nu(\rd \omega)$ for all $f\in (\Cc_h)_{++}(\Omega)$.
Let $\nu = \nu^r + \nu^s$ where $\nu^r$ is a measure with support in $\Omega$ and $\nu^s$ is a measure with support in $\breve{\Omega}\setminus \Omega$. 
For each $i\in \Ii$,
the extension $\breve{T}(\varphi_i)$ is continuous and hence by Assumption~\ref{assumption:functionH}$(3)$
we have that $\breve{T}(\varphi_i)(\omega) = 0$ for all $\breve{\Omega}\setminus \Omega$. 
Therefore we have
$$
\widetilde{\pi}\circ T(\varphi_i) = \int_{\breve{\Omega}}\breve{T}(\varphi_i)(\omega)\nu(\rd \omega)
= \int_{\breve{\Omega}}\breve{T}(\varphi_i)(\omega)\nu^r(\rd \omega) + 
 \int_{\breve{\Omega}\setminus \Omega}\breve{T}(\varphi_i)(\omega)\nu^s(\rd \omega)
 = \int_{\Omega}T(\varphi_i)(\omega)\nu^r(\rd \omega), 
$$
for all $i\in\Ii$. 
The last equality follows from the fact that the extension $\breve{T}(f)$ coincides with $T(f)$ on $\Omega$ for all $f\in \Cc_h(\Omega)$. 
Note also that $\nu^r\neq0$ otherwise one would have
$$
0<C\pi(m_0) = \widetilde{\pi}\circ T(m_0) = \int_{\breve{\Omega}}\breve{T}(m_0)(\omega)\nu^r(\rd \omega) + \int_{\breve{\Omega}\setminus \Omega}\breve{T}(m_0)(\omega)(\omega)\nu^s(\rd \omega) = 0, 
$$
where the last equality follows from the fact that $m_0\in o(h)$ and we arrive at 
a contradiction.
We can then define a probability measure on $\Omega$ as $\eta := \nu^r/\|\nu^r\|$ and 
$\widetilde{\pi}\circ T(f) = \int_{\Omega}f(\omega)\eta(\rd \omega)$,
for all $f\in \Cc_{b}(\Omega)$. 
Moreover defining the probability measure~$\mu$ via
$\frac{\rd\mu}{\rd\eta} := \frac{1}{h}\left(\int_{\Omega} \frac{1}{h(\omega)}\eta(\rd \omega)\right)^{-1}$ 
and setting $C:=\int_{\Omega} \frac{1}{h(\omega)}\eta(\rd \omega)$,
we see that $\mu\in\Pp_h(\Omega)$ and 
$\pi(g) = \langle g, \mu\rangle$, for any $g\in \Cc_{h}(\Omega)$.

\subsection{Proof of Theorem~\ref{thm:superReplication}}\label{app:thmsuperReplication}
We first prove the super-hedging case, and specialise to the case where $\Phi\in \Cc_{h}(\Omega)$.
Absence of weak free lunch and Assumption~\ref{assumption:functionH}
imply the existence of a Borel probability measure $\pi_0 \in\Pp_h(\Omega)$ that extends~$\barho$. 
It is clear that $\underline{\val}_p(\Phi) \leq \overline{\val}_p(\Phi)$.
If $\Phi\in\overline{\Mf}$ then $\overline{\val}_p(\Phi) = \underline{\val}_p(\Phi)$ and hence 
there is no duality gap between the primal~\eqref{eq:primal} and the dual~\eqref{eq:dual} programmes.
Assume $\Phi\notin\overline{\Mf}$ and fix some $\alpha\in(\pi_0(\Phi),\overline{\val}_p(\Phi))$.
Let $L:=\Span\{\overline{\Mf},\Phi\}\subset \Cc_{h}(\Omega)$, so that
any $l\in L$ can be represented as $l= m+\lambda \Phi$ for some $m\in\overline{\Mf}$ and $\lambda\in\RR$.
Define a functional $\eta : L \to \RR$ as 
$\eta(l) = \eta(m + \lambda\Phi) := \barho(m) + \lambda\alpha$.
It is linear and we now show that it is strictly positive on $L_{++} := L\cap(\Cc_{h})_{++}(\Omega)$. 
Let $z = m+\lambda\Phi \in L_{++}$ where $m\in\overline{\Mf}$ and $\lambda\in\RR$ and consider three cases.
If $\lambda=0$, then $\eta(z) = \barho(m) > 0$.
If $\lambda<0$, then $m>-\lambda\Phi$ and $\barho(m/(-\lambda)) \geq \overline{\val}_p(\Phi) > \alpha$ by assumption. 
Then $\eta(z) = \barho(m) + \lambda\alpha = -\lambda((-\lambda)^{-1}\barho(m) - \alpha) > 0$. 
Finally if $\lambda>0$, then $(-\lambda)^{-1}\barho(m) < \alpha$ and $\eta(z) >0$.

Introduce now the set $\Lr:=\{l\in L : \eta(l)\leq 0\}$, 
and note that $\Lr\cap (\Cc_{h})_+(\Omega) = \{0\}$ since~$\eta$ is strictly positive. 
We now show that $m_0\notin\overline{\Lr-(\Cc_{h})_+(\Omega)}$. 
Assume by contradiction that $m_0\in\overline{\Lr-(\Cc_{h})_+(\Omega)}$. 
Then there exists sequences $(f_n)_{n\in\NN}\subset \Cc_{h}(\Omega)$ converging to~$m_0$
and $(g_n)_{n\in\NN}\subset \Lr$ 
with $g_n = m_n + \lambda_n\Phi$ for $(m_n)_{n\in\NN}\subset\overline{\Mf}$, $(\lambda_n)_{n\in\NN}\subset\RR$ such that $g_n\geq f_n$ for all $n\in\NN$.
Clearly $m_n + \lambda_n\Phi - m_0 \geq f_n - m_0$ converges to zero, 
and hence 
$\liminf_n\eta(m_n + \lambda_n\Phi - m_0) \geq 0$ or equivalently 
$\limsup_n-\eta(g_n) + \barho(m_0) \leq 0$. 
Thus $0\geq\liminf_n\eta(g_n)\geq\barho(m_0) > 0$, which is a contradiction.
Therefore there exists a non-zero continuous linear functional $\pi: \Cc_{h}(\Omega) \to \RR$ 
such that $\pi(m_0)>0\geq \pi(g-f)$ for all $g\in \Lr$, $f\in(\Cc_{h})_+(\Omega)$
and by a similar argument to that used in the proof of Theorem~\ref{thm:FTAPnoOrderUnit},
$\pi$ extends~$\eta$, i.e. $\pi(l) = \eta(l) = \barho(m) + \lambda\alpha$ for all $l\in L$.
In particular $\pi$ extends $\barho$ and hence is a feasible solution to the dual programme~\eqref{eq:dual} and $\pi(\Phi) = \alpha$.
Moreover as $\pi$ is a feasible solution it follows that $\alpha\leq\overline{\val}_d(\Phi)$. 
As $\alpha\in(\pi_0(\Phi),\overline{\val}_p(\Phi))$ was chosen arbitrarily it implies 
that $\overline{\val}_d(\Phi) = \overline{\val}_p(\Phi)$.

Any $\Phi\in\Uu_h(\Omega)$ can be expressed as an infimum over continuous functions $(f_n)_{n\in\NN}$ 
that dominate it and, by Assumption~\ref{assumption:primalfeasibility} 
we can take them such that 
$\overline{\val}_p(f_n)<\infty$ for all $n\in\NN$.
As shown above, the no-duality gap holds for all $f\in \Cc_{h}(\Omega)$ with $\overline{\val}_p(f)<\infty$, 
and hence the duality result carries over to the upper semi-continuous case.

For the sub-hedging case, 
if $\Phi$ is lower semi-continuous then $-\Phi$ is upper semi-continuous 
and $\underline{\val}_p(\Phi) = -\overline{\val}_p(-\Phi)$, 
and the result follows by the Super-Replication  Theorem~\ref{thm:superReplication}.

\subsection{Proof of Lemma~\ref{lem:LinearVarNoArb}}\label{app:lemLinearVarNoArb}
Fix $a_0\in\RR_+$ and $a_1\in[0,2]$.
If $a_1 = 0$ then $w(k,t) = a_0$ for all $k\in\RR$ and~$\gr$ is constant equal to $1$.
We thus assume $a_1\in(0,2]$.
Since $w(\cdot,t)$ is linear, the function~$\gr$ reads
\begin{equation}
\gr(k)= \left(1 - \frac{a_1k}{2w(k,t)}\right)^2 - \frac{a_1^2}{4}\left(\frac{1}{w(k,t)} + \frac{1}{4}\right)
= \left(\frac{w(k,t) + a_0}{2w(k,t)}\right)^2 - \frac{a_1^2}{4}\left(\frac{4+w(k,t)}{4w(k,t)}\right).
\end{equation}
Let us denote $x := w(k,t)$.
Then the above expression becomes
\begin{equation}
\gr\left(\frac{x-a_0}{a_1}\right)= 
\frac{1}{16x^2}\left(4x^2 + 8a_0x + 4a_0^2 - 4a_1^2x - a_1^2x^2 \right)
= \frac{(4 - a_1^2)x^2 + 4(2a_0 - a_1^2)x + 4a_0^2}{16x^2}.
\end{equation}
If $a_1=2$ then the numerator is linear in~$x$. 
Solving for~$x$ yields the root 
$x = \frac{-4a_0^2}{8(a_0-2)}$.
Clearly~$\gr$ is non-negative for $a_0<2$ 
and substituting~$k$ back produces the expression
\begin{equation}\label{eq:k_func1}
k^*(a_0,a_1) = \frac{a_0(8-6a_0)}{8(a_0-2)},
\end{equation}
which is positive if $a_0\in(4/3,2)$. 
Consider now the case when $a_1\in(0,2)$. 
The numerator in the expression for~$\gr$ above is quadratic in~$x$, 
and solving for~$x$ yields two roots 
$$
x_{\pm} = \frac{-2(2a_0 - a_1^2) \pm 2a_1\sqrt{a_0^2 - 4a_0 + a_1^2}}{4 - a_1^2}.
$$
As $x = a_1k + a_0$ the corresponding values of $k$ are 
\begin{equation}\label{eq:k_func2}
k_{\pm} = \frac{a_1(a_0 + 2) - \frac{8a_0}{a_1} \pm 2\sqrt{a_0^2 - 4a_0 + a_1^2}}{4 - a_1^2},
\end{equation}
and both roots are real if and only if 
$a_0 \in \RR\setminus (2 - \sqrt{4 - a_1^2}, 2 + \sqrt{4 - a_1^2})$ for $a_1\in(0,2]$.
If~$a_0\geq 2-\sqrt{4 - a_1^2}$ then substituting the lower bound for $a_0$ into the expression for $\gr$ above we get 
\begin{align*}
\gr\left(\frac{x-a_0}{a_1}\right) &\geq 
\frac{(4 - a_1^2)x^2 + 4(4 - 2\sqrt{4 - a_1^2} - a_1^2)x + 4(4-4\sqrt{4 - a_1^2} + 4 - a_1^2)}{16x^2}
\\
&= \frac{(4 - a_1^2)(x + 2)^2 - 8x\sqrt{4 - a_1^2} + 16}{16x^2}= 
\frac{(x\sqrt{4 - a_1^2} - 4)^2}{16x^2} \geq 0,
\end{align*}	
for all $k>0$.
On the other hand if $a_0 <  2 - \sqrt{4 - a_1^2}$ then $\gr$ is strictly positive for all $k>k_+$ and setting $k^*(a_0,a_1) = k_+$ we obtain the result.

Suppose now that $\gr(k)\geq 0$ for all $k\in [k^*(a_0,a_1),\infty)$. 
The second derivative of the Black-Scholes formula with respect to $e^k$ gives for any $k\in [k^*(a_0,a_1),\infty)$ the Call price $c(k,t)$ expressed as
\begin{equation}
c(k,t) = \frac{\gr(k)}{\sqrt{2\pi w(k,t)}}\exp\left(-\frac{\left(d(k,\sqrt{w(k,t)}) - \sqrt{w(k,t)}\right)^2}{2}\right),
\end{equation}
which is non-negative by assumption on $\gr$. 
As $k\uparrow\infty$ by assumption we have that $w(k,t) \sim a_1k$ and note that 
$d(k,\sqrt{a_1k}) - \sqrt{a_1k} = -(1/\sqrt{a_1} + \sqrt{a_1}/2)\sqrt{k}$. 
Recalling the bound~\cite[Theorem~2.1]{Lee} that holds for all $p\geq 0$ (with $p=0$ being the trivial bound),
we obtain that $a_1 = \psi(p)$, i.e. $a_1\in [0,2]$.


\subsection{Proof of Theorem~\ref{thm:Convergence}}\label{app:thmConvergence}
We start with the convergence of the sets of martingale measures:
\begin{lemma}\label{lemma:CountableMgCondition}
Let $r:=\min\{d(t): t\in\Tt\}$. 
As~$r$ tends to infinity, the set $\widetilde{\MMM}^{p^*,q^*}_{\Cb}$ converges to the set of martingale measures consistent 
with the traded Call option prices~$\Cb$. 
\end{lemma}
\begin{proof}
It is sufficient to show that the limit of sets $\widetilde{\MMM}^{p^*,q^*}_r$ defined as 
\begin{equation}
\widetilde{\MMM}^{p^*,q^*}_r := \left\{ \mu \in \Pp_h(\Omega) : \int_{\Omega}\left(\Theta \bullet S(\omega)\right)_T\mu(\rd \omega) = 0, \text{ for all } \Theta\in\widetilde{\Hh}\right\}, 
\end{equation} 
the set of probability measures in $\Pp_h(\Omega)$ that integrate $(\Theta \bullet S)_T$ to zero for all $\Theta\in\widetilde{\Hh}$ (where $\widetilde{\Hh}$ is dependent on $r$ via the choice of $l(t_j)$, $j=1,\ldots,n-1$),
converges to the set of martingale measures $\MMM^{p^*,q^*}$ defined in~\eqref{eq:martingaleHigherMoments}
For any $j=1,\ldots,n-1$, define the set $B^{\infty}_j := \lim_{l(t_j)\uparrow\infty}B_j$ and let $\widetilde{\Hh}_{\infty} := \RR \times \prod_{j=1}^{n-1} B^{\infty}_j$. 
It is clear that $B = \cup_{j=1}^{n-1}B^{\infty}_j$.
Define further the limit $\widetilde{\MMM}^{p^*,q^*}_{\infty} := \lim_{r\uparrow\infty}\widetilde{\MMM}^{p^*,q^*}_{r}$.
It is clear that $\MMM^{p^*,q^*}\subseteq \widetilde{\MMM}^{p^*,q^*}_{\infty}$.
To show the reverse inclusion define the gain of the trading strategy $\Theta \in \widetilde{\Hh}$ at time $t\leq T$ as 
$$
\left(\Theta\bullet S\right)_t := a_0(S_{t_1} - s_0) + \sum_{j=1}^{\max\{k: t<t_k\in\Tt\}}\sum_{i=1}^{l(t_j)}a^{t_j}_i\theta^{t_j}_i\left(S_{t_{j+1}} - S_{t_j}\right).
$$
Introduce the stopping time $\tau_{\alpha} := \min\{t\in\Tt: S_{t} > \alpha\}$.
The set $\widetilde{\MMM}^{p^*,q^*}_{\infty}$ consists of all measures $\mu\in\Pp_h(\Omega)$ such that $\langle (\Theta\bullet S)_{T\wedge \tau_{\alpha}},\pi\rangle = 0$ for each $\alpha\in\mathbb{Q}$. 
By definition of the set~$B$, for each $\alpha$ and $j$ any function $f\in \Cc_b(K^j_{\alpha})$---in 
particular the indicator function~$\ind_{K^j_{\alpha}}$---can be approximated by elements in~$B$;
hence $\langle (\Theta_0 \bullet S)_{T\wedge\tau_{\alpha}},\pi\rangle = 0$, 
where $\Theta_0 := (a_0,\ind_{K^1_{\alpha}},\ldots,\ind_{K^{n-1}_{\alpha}})$ and as for all $\alpha\in\mathbb{Q}$ and each $j$ the sets $K^j_{\alpha}$ generate Borel sigma algebra on $\RR_+^j$ it follows that $S_{T\wedge\tau_{\alpha}}$ is a martingale under~$\pi$, 
and therefore $(S_{t_j})_{j=1,\ldots,n}$ is a $\pi$-local martingale. 
Since~$S_T$ is integrable with respect to any $\mu\in\Pp_h(\Omega)$ it follows from~\cite[Theorem~2(b)]{JacodShiryaev} that it is a martingale under any $\pi\in\widetilde{\MMM}^{p^*,q^*}_{\infty}$, 
and hence $\widetilde{\MMM}^{p^*,q^*}_{\infty}\subseteq\MMM^{p^*,q^*}_{\infty}$.
\end{proof}

As in the proof of Lemma~\ref{lemma:CountableMgCondition}, let $\widetilde{\Hh}_{\infty}$ be a countable subset of~$\Hh$. 
The sequence of nested sets $(\widetilde{\Hh}_r)_{r\in\NN}$ 
with $\widetilde{\Hh}_r\subset\widetilde{\Hh}_{r+1}$ represents the discretised trading strategies as the bases~$B_j$ increase for each $j=1,\ldots,n-1$ simultaneously, 
and clearly
$\widetilde{\Hh}_{\infty} = \lim_{r\uparrow\infty}\widetilde{\Hh}_r$. 
For any $r\in\NN$, let~$\overline{\val}^r_p(\Phi)$ be the primal problem~\eqref{eq:primalSemiInfinite} over the set of primal variables $\RR^{\df+1}\times\widetilde{\Hh}_r$.
Likewise, 
we denote $\overline{\val}^r_d(\Phi)$ the dual problem~\eqref{eq:dualSemiInfinite} over the set of probability measures in $\Pp_h(\Omega)$ that re-price given Call options~$\CC$ and satisfy the martingale condition for all $\Theta\in\widetilde{\Hh}_r$.
By assumption there is no duality gap between the primal and the dual problems, i.e. $\overline{\val}^r_p(\Phi) = \overline{\val}^r_d(\Phi)$ for all $r\in\NN$;
since both sequences $(\overline{\val}^r_p(\Phi))_{r\in\NN}$ and $(\overline{\val}^r_d(\Phi))_{r\in\NN}$ are non-increasing, their limits exist and $\lim_{r\uparrow\infty}\overline{\val}^r_p(\Phi) = \lim_{r\uparrow\infty}\overline{\val}^r_d(\Phi)$. 
We also define $\overline{\val}^{\infty}_d(\Phi) := \lim_{r\uparrow\infty}\overline{\val}^r_d(\Phi) > -\infty$
with  
\begin{equation}
\overline{\val}^{\infty}_d(\Phi) := \sup\left\{\int_{\Omega}\Phi(\omega)\mu(\rd\omega) : \mu\in\widetilde{\MMM}^{p^*,q^*}_\infty\text{, } 
\int_{\Omega}\CC(\omega)\mu(\rd\omega) = \cc
\right\},
\end{equation}
where the set $\widetilde{\MMM}^{p^*,q^*}_\infty$ is defined in the proof of Lemma~\ref{lemma:CountableMgCondition}.
Therefore the value of the semi-infinite dual problem~\eqref{eq:dualSemiInfinite} 
converges to the value of the infinite-dimensional dual problem~\eqref{eq:dualUpperSemiContWeakArbCh2}
by Lemma~\ref{lemma:CountableMgCondition}. 
It follows that the value of the semi-infinite primal problem~\eqref{eq:primalSemiInfinite} 
also converges to the value of the infinite-dimensional primal problem~\eqref{eq:primalUpperSemiContWeakArbCh2}
as $\lim_{r\uparrow\infty}\overline{\val}^r_p(\Phi) = \lim_{r\uparrow\infty}\overline{\val}^r_d(\Phi)$ 
and there is no duality gap between the infinite-dimensional primal~\eqref{eq:primalUpperSemiContWeakArbCh2} and the dual~\eqref{eq:dualUpperSemiContWeakArbCh2} problems.


\bibliographystyle{apa}
\bibliography{bibliography} 
\end{document}